\def\year{2022}\relax
\documentclass[letterpaper]{article} 
\usepackage{aaai22}  
\usepackage{times}  
\usepackage{helvet}  
\usepackage{courier}  
\usepackage[hyphens]{url}  
\usepackage{graphicx,booktabs} 
\usepackage[toc,page]{appendix}
\usepackage{natbib}  
\usepackage{caption} 
\DeclareCaptionStyle{ruled}{labelfont=normalfont,labelsep=colon,strut=off} 
\usepackage{algorithm}

\usepackage{newfloat}
\usepackage{listings}
\floatstyle{ruled}
\newfloat{listing}{tb}{lst}{}
\floatname{listing}{Listing}
\title{FairFoody: Bringing in Fairness in Food Delivery}
\author {
Anjali Gupta,
Rahul Yadav,
Ashish Nair,\\
Abhijnan Chakraborty,
Sayan Ranu,
Amitabha Bagchi
}
\affiliations {
Department of Computer Science and Engineering \\
Indian Institute of Technology, Delhi\\
\{anjali,Rahul.Yadav.cs517,Ashish.R.Nair.cs517,abhijnan,sayanranu,bagchi\}@cse.iitd.ac.in
}

\usepackage{bibentry}

\usepackage{cite,url,xspace}
\usepackage{amsmath,amssymb,amsfonts}
\usepackage{graphicx}
\usepackage{textcomp}
\usepackage{xcolor}
\usepackage{comment}
\usepackage{amsthm}
\usepackage{wasysym}
\usepackage{multirow}
\usepackage{algorithm}
\usepackage[noend]{algpseudocode}
\algblockdefx[Foreach]{Foreach}{EndForeach}[1]{\textbf{for each} #1 \textbf{do}}{\textbf{end for}}
\usepackage[bf]{subfigure}

\usepackage{enumitem}
\usepackage{times}
\setlist{nolistsep,leftmargin=*}

\def\fmplus{\textsc{FoodMatch}\xspace}
\def\fmplusone{\textsc{Food-}\xspace}
\def\fmplustwo{\textsc{Match}\xspace}
\def\fair{\textsc{FairFoody}\xspace}
\def\fairone{\textsc{Fair-}\xspace}
\def\fairtwo{\textsc{Foody}\xspace}
\def\abhi{\textsc{2sf}\xspace}

\algnewcommand{\IIf}[1]{\State\algorithmicif\ #1\ \algorithmicthen}
\algnewcommand{\EndIIf}{\unskip\ \algorithmicend\ \algorithmicif}

\newtheorem{theorem}{\textbf{Theorem}}
\usepackage{amsthm}
\newtheorem*{theorem2*}{Theorem~2}
\newtheorem*{theorem3*}{Theorem~3}
\newcommand{\updatetext}[1]{\textcolor{black}{#1}}
\newtheorem{definition}{\textbf{Definition}}

\newtheorem{example}{\textbf{Example}}
\newtheorem{lemma}{\textbf{Lemma}}

\newtheorem{corollary}{\textbf{Corollary}}
\newtheorem{problem}{\textbf{Problem}}

\def\CV{\mathcal{V}}
\def\CR{\mathcal{R}}
\def\CL{\mathcal{L}}
\def\cc{\mathcal{C}}

\def\maxorders{\text{MaxO}}
\def\waitTime{\text{wT}}
\def\driveTime{\text{dT}}
\def\availTime{\text{aT}}
\def\inc{\text{inc}}
\def\ninc{\text{ns-inc}}

\begin{document}

\maketitle

\begin{abstract}
Along with the rapid growth and rise to prominence of food delivery platforms, concerns have also risen about the terms of employment of the ``gig workers'' underpinning this growth. Our analysis on data derived from a real-world food delivery platform across three large cities from India show that there is significant inequality in the money delivery agents earn. In this paper, we formulate the problem of fair income distribution among agents while also ensuring timely food delivery. We establish that the problem is not only \textit{NP-hard} but also \textit{inapproximable in polynomial time}. We overcome this computational bottleneck through a novel matching algorithm called \fair.  
Extensive experiments over real-world food delivery datasets show \fair imparts up to $10$ times improvement in equitable income distribution when compared to baseline strategies, while also ensuring minimal impact on customer experience.

\end{abstract}

\section{Introduction}
\noindent
Food delivery platforms like 
DoorDash, Zomato, GrubHub, Swiggy and Lieferando have become popular means for people to order food online and get delivery at their doorsteps, increasingly so due to the covid-19 pandemic related restrictions~\cite{fooddelivery-covid}. Typically, when a customer orders food from a particular restaurant a delivery agent is assigned to pick up the food from the restaurant once ready and deliver it to the customer. Thus, apart from providing business opportunities to the restaurants, food delivery platforms also provide livelihood to thousands of delivery agents.
However, recent media reports have highlighted a range of issues faced by  
these delivery agents: poor working conditions, long working hours, non-transparent job allocations, and meagre pay~\cite{delivery-agent-plight1,delivery-agent-plight2,fooddelivery-newsminute}. Due to the `gig' nature of delivery jobs, a delivery agent typically gets a small fixed commission per order (except occasional tips and other incentives), and few employment benefits. A recent survey by the non-profit Fairwork~\cite{fairwork} found that despite the gig labels associated with this work (denoting part time engagements in addition to more stable jobs), most delivery agents in developing countries actually work full time on these platforms, depending on them entirely for their livelihood~\cite{fairwork}. Fairwork also found that none of the food delivery platforms in India guarantee local minimum wage to the delivery agents even if they work for more than 10 hours per day~\cite{fairwork}. 

Increasing the pay of delivery agents is a complicated proposition. If we decrease the number of agents so that the per-agent pay increases, there is a danger of increased customer wait time, which is anathema in the food delivery sector. Charging higher commission from restaurants or offering higher pay-per-delivery may disproportionately affect smaller restaurants and decrease their customer base. Recently such concerns led the city of Chicago to cap the delivery charge at $10\%$ of the order value~\cite{chicago-fee-cap}. In some cities, restaurants are offering their own delivery services to cut back on the platform charges~\cite{fooddelivery-restaurant}. In summary, it is not easy to increase 
the pool of money available to remunerate the delivery agents. However, apart from the issue of a small pool of money, the Fairwork report suggests that there is also high variability in pay -- earnings on a platform can vary widely across agents 
~\cite{fairwork}. 
To investigate this issue further, we 
perform an in-depth analysis of data obtained from a major food delivery platform for three Indian cities. Our analysis shows that there is significant inequality in the amount of money different agents earn from the platform. Interestingly, we see that the number of working hours of an agent, or 
her hours of operation cannot explain this inequality -- it is rather the catchment area which makes a major difference. 

In this paper, we present \fair, the first algorithm to address the food delivery order allocation problem with fairness of pay distribution as a goal.
We address the issue of catchment-based inequality by doing away with the restriction that a delivery agent must work within a single zone. Our algorithm, \fair, does not, however, make any special effort to move agents across zones. In fact, in order to ensure timely delivery it limits the range from which a delivery agent can be picked to deliver a particular order. But, nonetheless, 
to ensure a more equitable pay distribution, \fair ends up creating a more uniform geographical distribution of agent activity. 

The challenge we faced in designing \fair is that at any given time the number of orders may not be large enough to ensure that every idle agent is kept busy, and, hence, fairly remunerated. To deal with this issue we drew on Fairwork's finding that a number of delivery agents treat food delivery as a full-time job to relax this temporal constraint: since the orders arrive throughout the day, we {\it amortize fairness over a longer period of time than trying to be fair at each assignment}. In our scheme if an agent does not get her fair share of assignments in a given 
time period, she may still make up for it on subsequent periods, 
and get a fair 
income over a longer term. With this relaxation \fair is able to fairly remunerate those delivery agents who spend a significant period of time on the platform, and, consequently, rely on it to provide them with a living wage. 

In summary, 
our key contributions are as follows:

\begin{itemize}
\item \textbf{In-depth investigation of food-delivery data: } We perform the first in-depth study of real food delivery data from large metropolitan cities 
and establish that the income distribution, even when normalized by the active hours of a delivery agent, shows high levels of inequality. 

\item \textbf{Problem formulation and algorithm design:} We formulate the \textit{multi-objective optimization} problem of fair income distribution in food delivery assignment, without compromising on customer experience. We show that the problem is not only \textit{NP-hard}, but also \textit{inapproximable } in polynomial time. To mitigate this computational bottleneck, we develop an algorithm called \fair that uses \textit{bipartite matching} on a data stream to perform real-time fair assignment of orders. 
ours is the first 
proposal to ensure fairness in food-delivery.

\item \textbf{Empirical evaluation:} We evaluate \fair extensively on real food delivery data across a range of metrics, and establish that it is successful in its dual objective of fair income distribution and positive customer experience.
\end{itemize}

\begin{figure*}[t]
\centering  
\subfigure[{Payment}]{\label{fig:lorenz_pay} \includegraphics[width=1.33in]{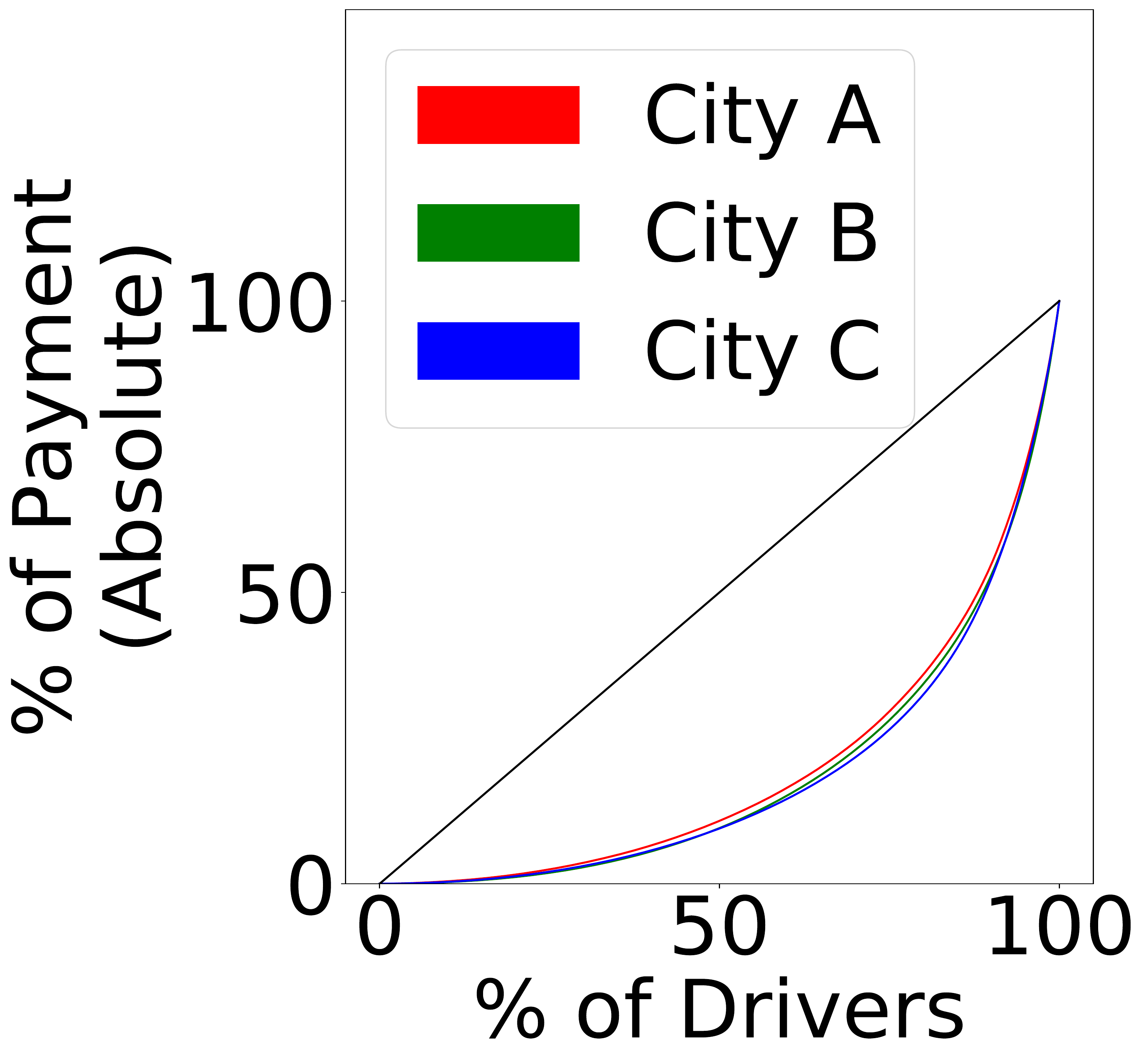}\label{fig:foodmatch_gini}}
\hfill
\subfigure[{ Payment/hour}]{\label{fig:lorenz_normalized_pay} \includegraphics[width=1.33in]{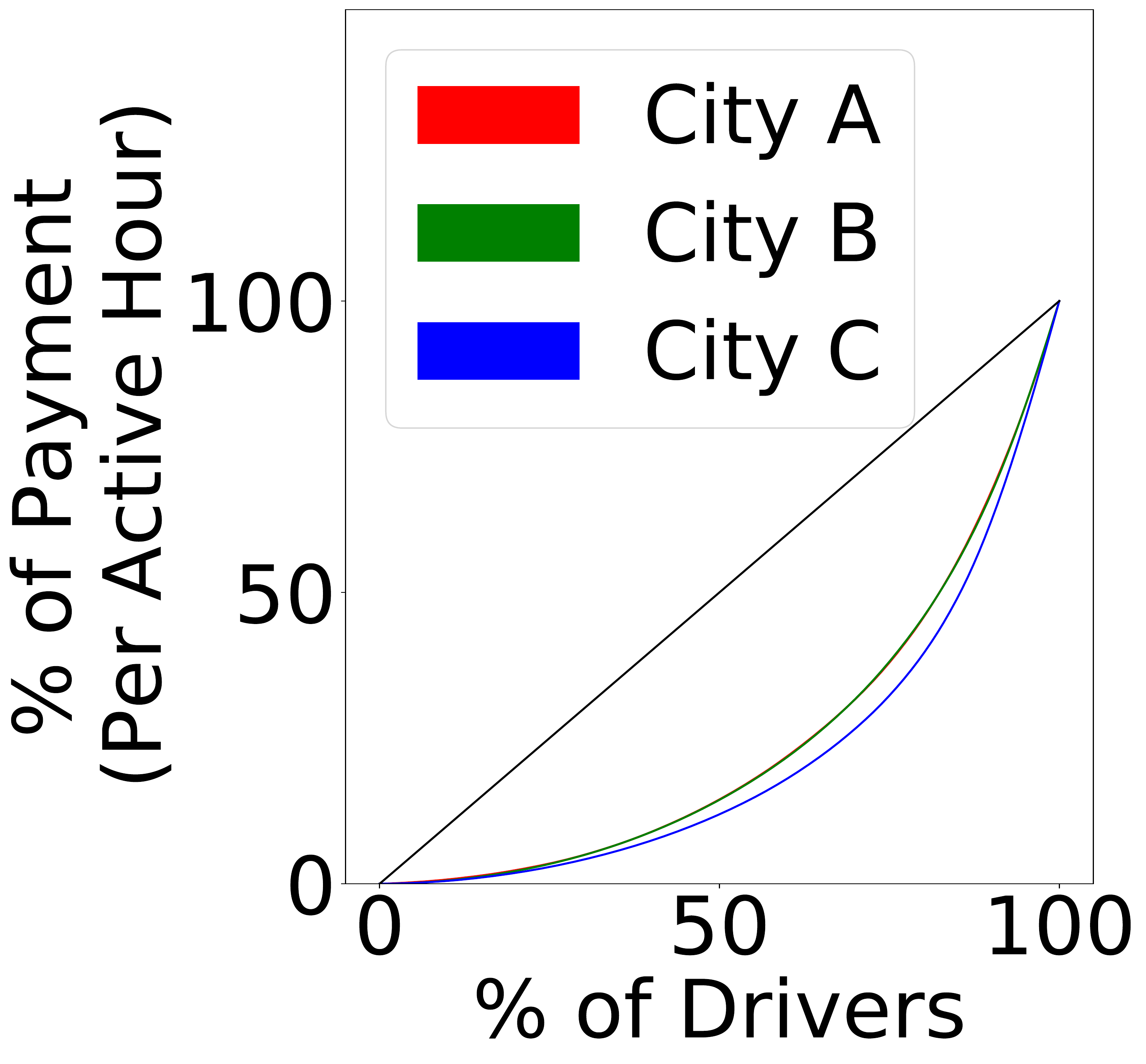}\label{fig:foodmatch_gini_norm}}
\subfigure[{ City A}]{\includegraphics[width=1.33in]{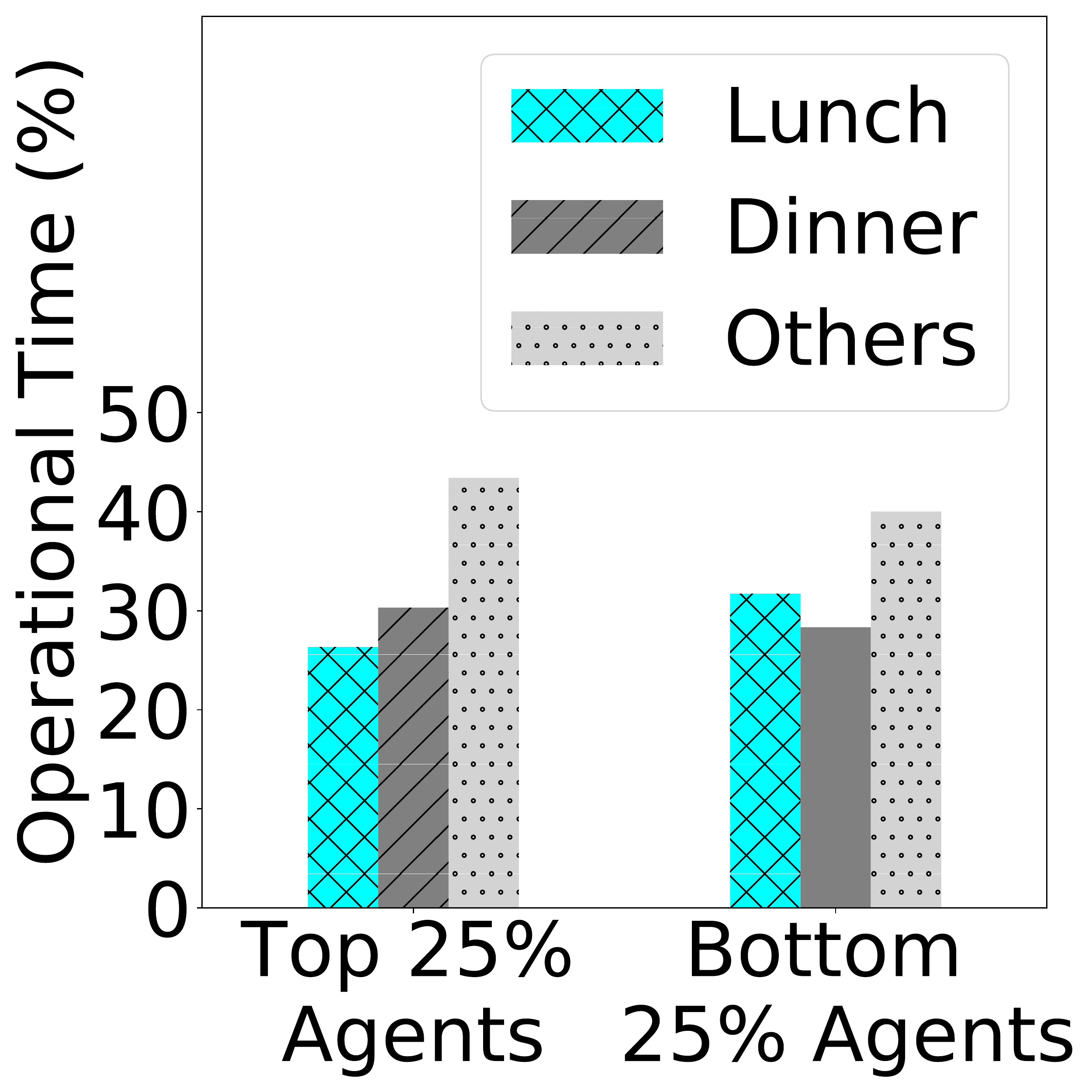}}\hfill
\subfigure[{ City B}]{\includegraphics[width=1.33in]{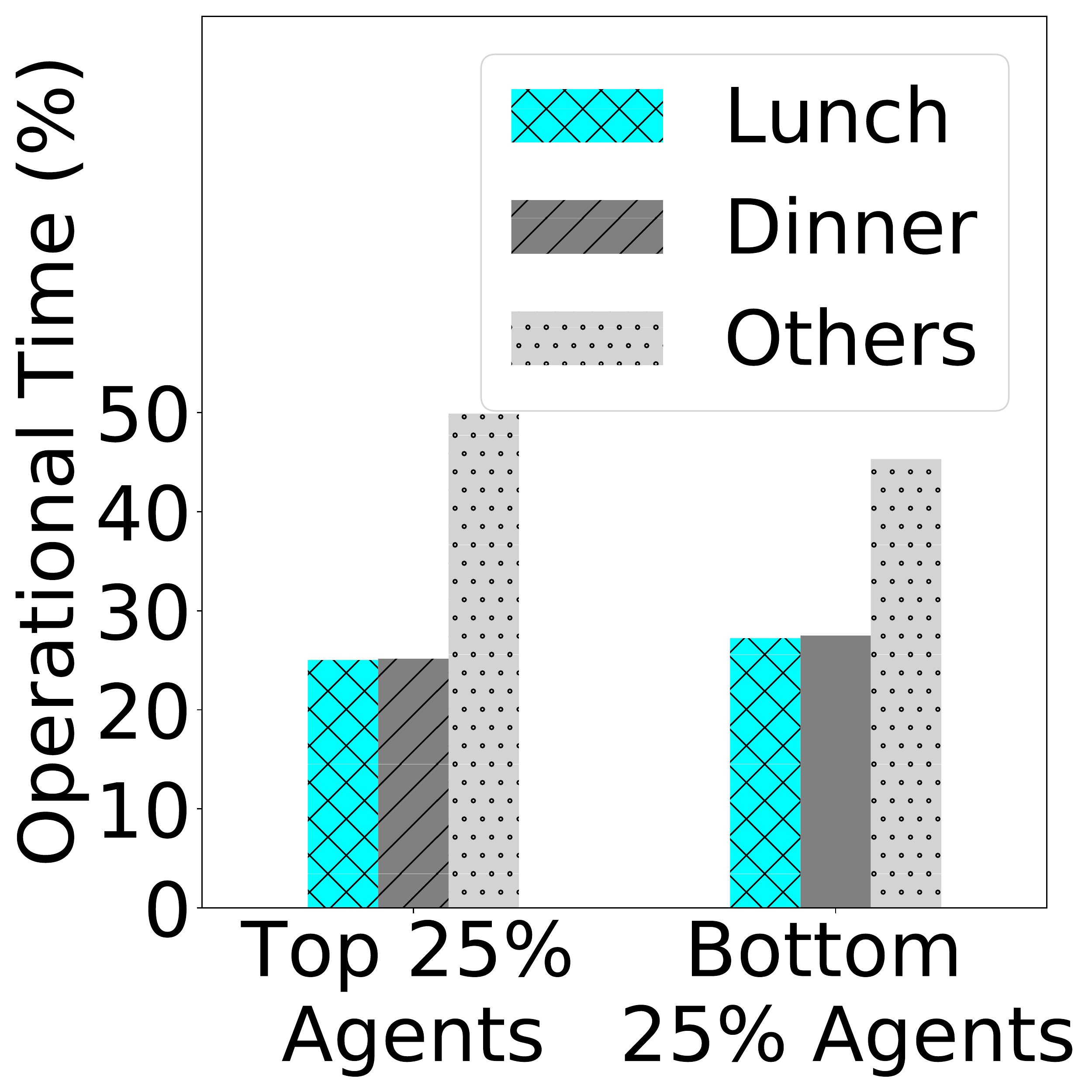}}\hfill
\subfigure[{ City C}]{\includegraphics[width=1.33in]{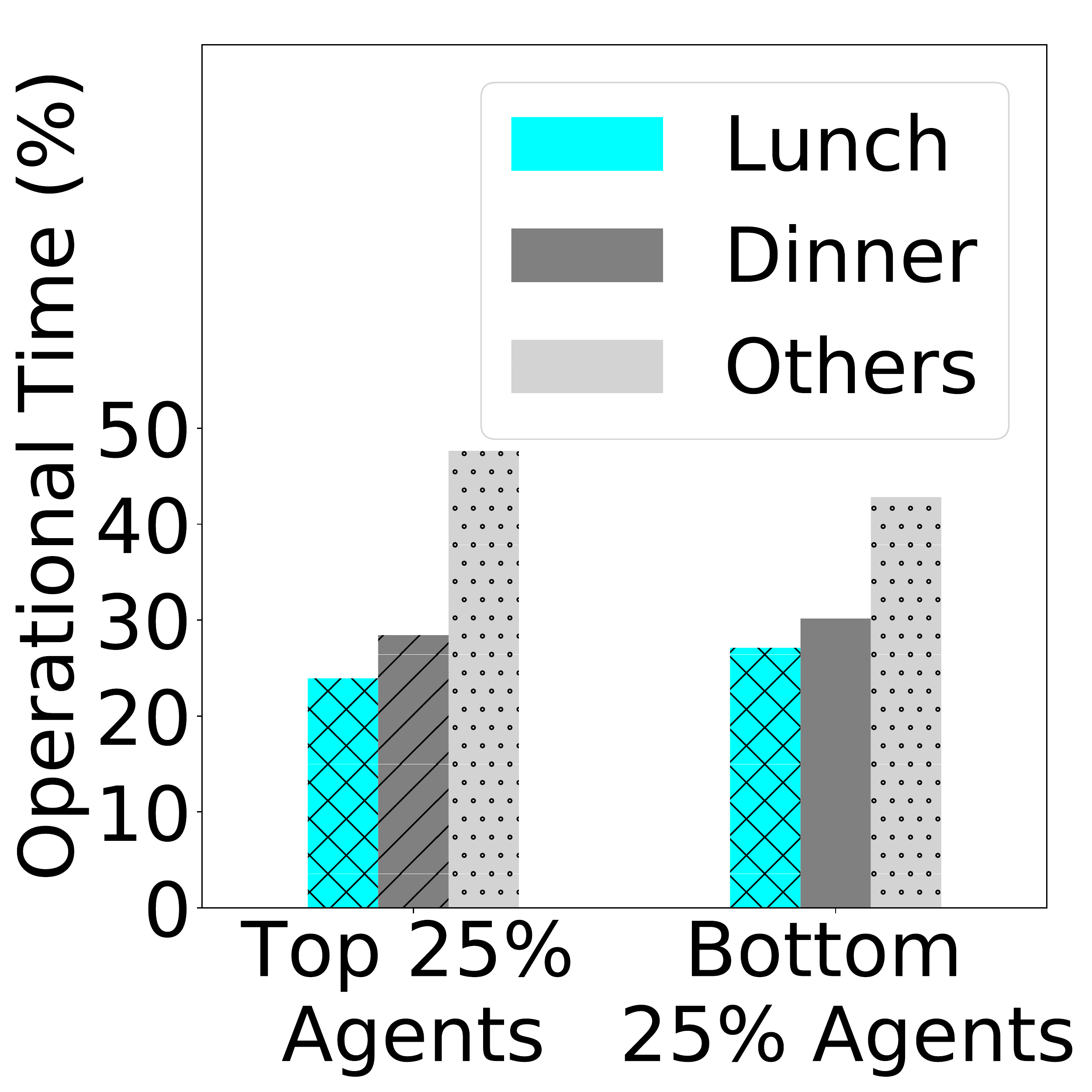}}
\caption{(a-b) Lorenz curves for the distribution of agent incomes. (c-e) Distribution of operational periods for top earners and bottom earners.}
\label{fig:op_time_dist}
\end{figure*}

\section{Inequality in Delivery Agents’ Income}
\subsection{Dataset Used}
\label{sec:dataset}
\noindent
We use six days of food delivery data from three large Indian cities, 
provided to us by 
a major food delivery service provider in India. 
The six days span Wednesday to Monday and therefore includes both weekdays and weekends.Table~\ref{tab:dataset} summarizes the dataset characteristics. The dataset consists of three components: trajectories of the delivery vehicles, the road network of each city (obtained from {\tt OpenStreetMap.org}), and metadata describing various factors such as the vehicle IDs, information 
on each received order, locations of restaurants and customers, mean food preparation time in each restaurant, average speed in each road segment at different hours, 
etc. 
We match the vehicle GPS pings to the road network 
to obtain network-aligned trajectories~\cite{mapmatch}. 

Each of these three cities provides different flavors of customer behavior. 
While City B and City C are large metropolitan cities with more than 8 million residents, City A is a  
comparatively smaller city with 5 million inhabitants and lower order volume. 
Although City C has a larger number of restaurants, $41\%$ more orders were fulfilled in City B in the time period under consideration. Furthermore, $27\%$ more vehicles were employed in City B to cope with the higher order volume. 

\subsection{Inequality in Payment Distribution}
\noindent
We perform allocation using \fmplus~\cite{foodmatch_arxiv} on the described dataset and compute the payment earned by different delivery agents over the course of these $6$ days.  Fig.~\ref{fig:lorenz_pay} shows the Lorenz curve of their income, where the $y$-axis represents the cumulative percentage of total income and the $x$-axis represents the cumulative percentage of all agents. The diagonal line in Fig.~\ref{fig:lorenz_pay} is the equality line; the further the income distribution from this line, the higher is the inequality. We see in Fig.~\ref{fig:lorenz_pay} that there is high inequality in the income earned by the agents; the top $10\%$ earners get $50\%$ of the total payment while the bottom $60\%$ only get $10\%$ of it. Such a high inequality can starve many agents from getting decent income and would force them to quit the platform.

\begin{table}[t]
\centering
\scalebox{0.9} {
\begin{tabular}{lccc} 
    \hline
    {\bf } & {\bf City A} & {\bf City B} & {\bf City C}\\ \hline 
    {\bf\# Restaurants} & $2085$ & $6777$ & $8116$\\
    {\bf\# Vehicles} & $2454$ & $159160$ & $10608$\\ 
    {\bf\# Orders} & $23442$ & $112745$ & $112745$\\
    {\bf Food prep. time} (avg.in min) & $8.45$ & $9.34$ & $10.22$\\
    {\bf\# Nodes} & 39k & 116k & 183k\\
    {\bf\# Edges} & 97k & 299k & 460K\\
 \hline
\end{tabular}}
\caption{Summary of the dataset.}
\label{tab:dataset}
\end{table}

\subsection{What Drives Inequality?}
\noindent
Next, we try to uncover the sources of such high inequality. 

\noindent
\textbf{I. Number of working hours: } It is reasonable to assume that pay variability springs from variability in the number of hours worked. However this is not the case in our data. We normalized payment by the number of active hours worked: Fig.~\ref{fig:lorenz_normalized_pay} shows the Lorenz curve of the hourly income. We can see in Fig.~\ref{fig:lorenz_normalized_pay} that the high inequality persists even after accounting for the activity levels.

\noindent
\textbf{II. Hours of operation: } We may conjecture that 
top earners are active during lunch (11AM-2PM) or dinner (7PM-11PM) times when order volumes are high, and the bottom earners are active during other periods of the day. Fig.~\ref{fig:op_time_dist}(c,d,e) show the distribution of the operational periods (fraction of active times during lunch, dinner and all other time periods) for both the top $25\%$ and bottom $25\%$ earners. We can see that except City A, there is no noticeable difference in their activity patterns. It is not 
that the top earners were overwhelmingly more active during lunch and dinner times compared to the bottom earners. In fact, in City A, we can see that the bottom earners were more active during lunch times. Thus, difference in activity period is not a reasonable explanation for the inequality.

\noindent
\textbf{III. Geographical distribution: } Next, we focus on the geographical spread of the order locations. Fig.~\ref{fig:foodmatch_heatmap} show heatmaps of restaurants' and customers' locations for the orders assigned to top earners and bottom earners, along with the location of the delivery agents when the orders were assigned to them. We can see a clear difference in the areas where top and bottom earners were active. For example, Fig.~\ref{fig:foodmatch_heatmap} shows that top earners deliver food mostly in the western part of City C, whereas the bottom earners are active more in the eastern part. It turns out that the order volume is higher in the western part, creating the inequality. We see similar trend in City A and City B as well. 

Our finding here corroborates the delivery agents' experience as reported by Fairwork: agents often complained about not receiving orders in areas other than their chosen pick-up zone~\cite{fairwork}. Even when an agent delivers food outside their zone, they are not allocated any orders on their way back, thus incurring fuel costs on their return journey, without getting any payment~\cite{fairwork}. We suggest that the delivery platforms can distribute opportunities more fairly among agents by allowing them to deliver orders in different parts of the city, as long as it does not negatively impact the waiting time for customers.

\begin{figure}[t]
\centering  
\subfigure[][Customers ]{\includegraphics[width=1.2in]{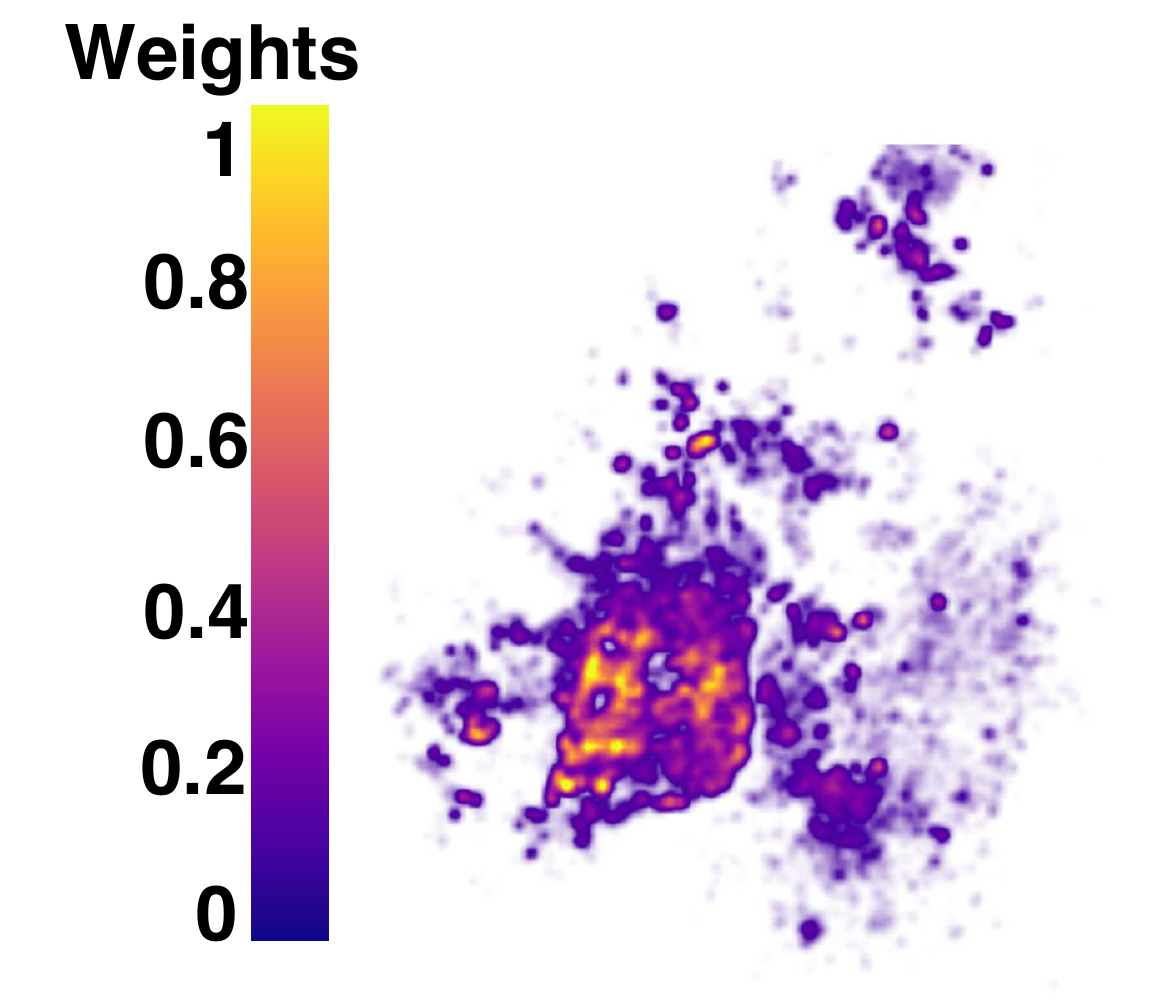}}
\subfigure[][Restaurants]{\includegraphics[width=1in]{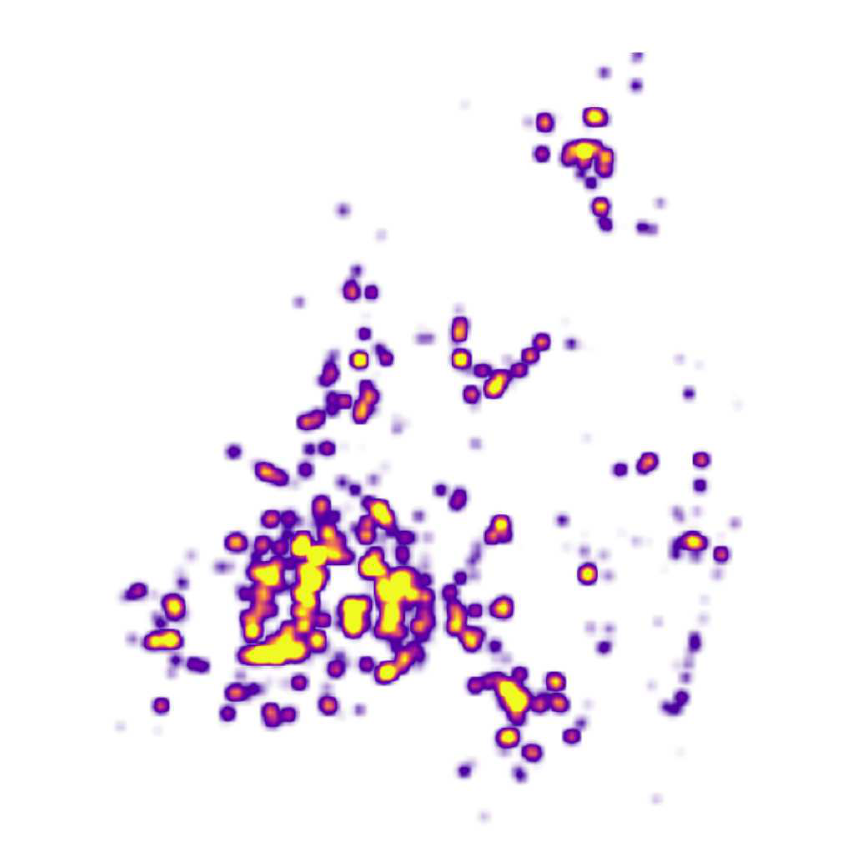}}
\subfigure[][Delivery Agents]{\includegraphics[width=1in]{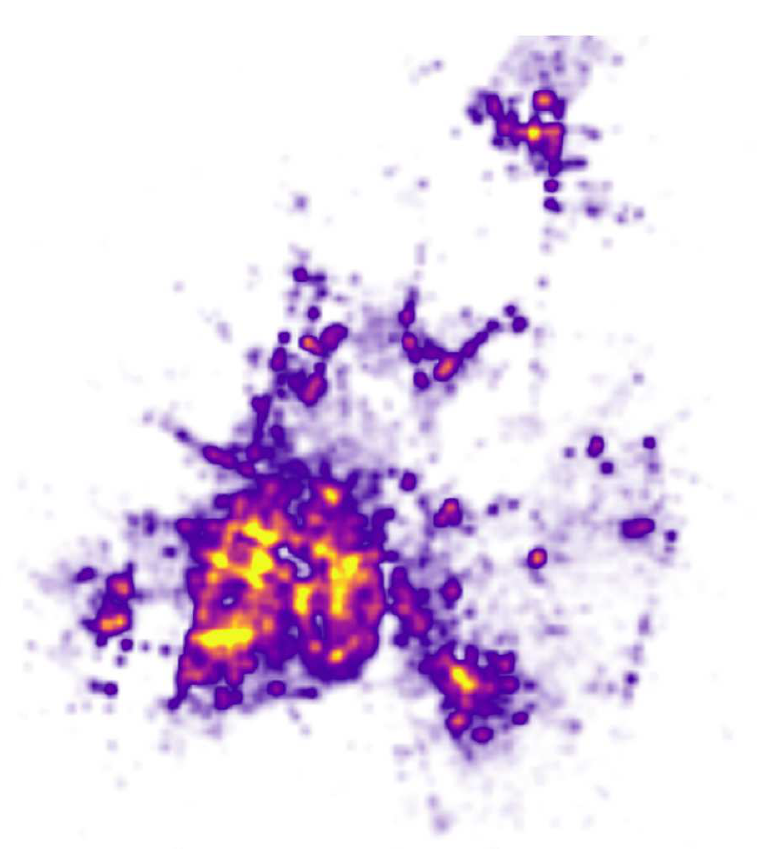}}\\
\subfigure[][Customers]{\includegraphics[width=1in]{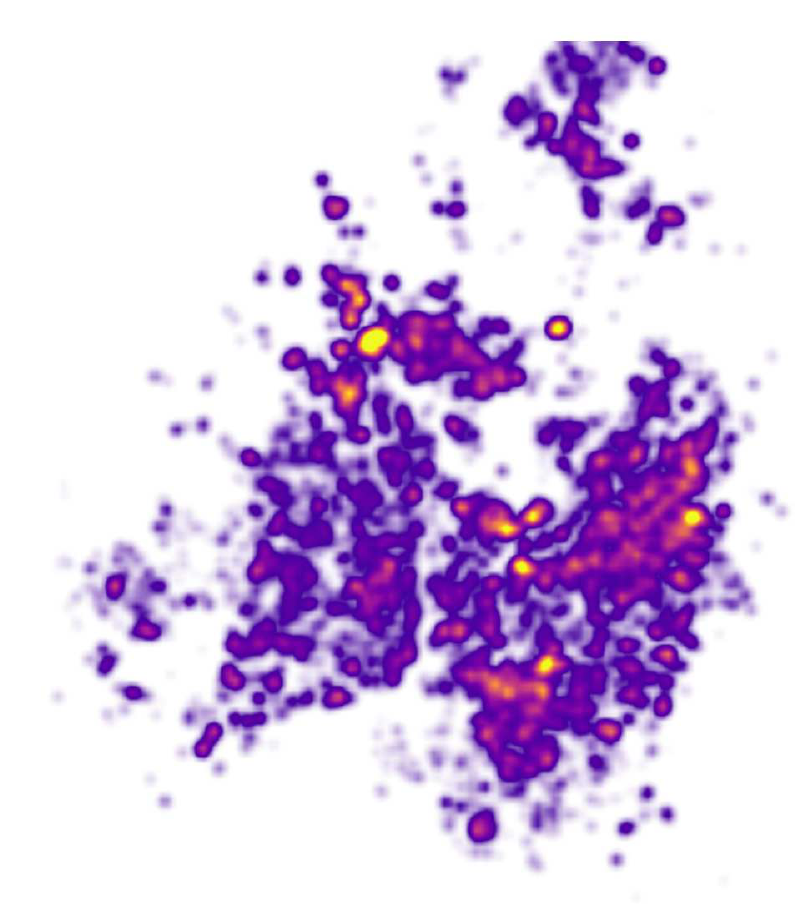}}
\subfigure[][Restaurants]{\includegraphics[width=1in]{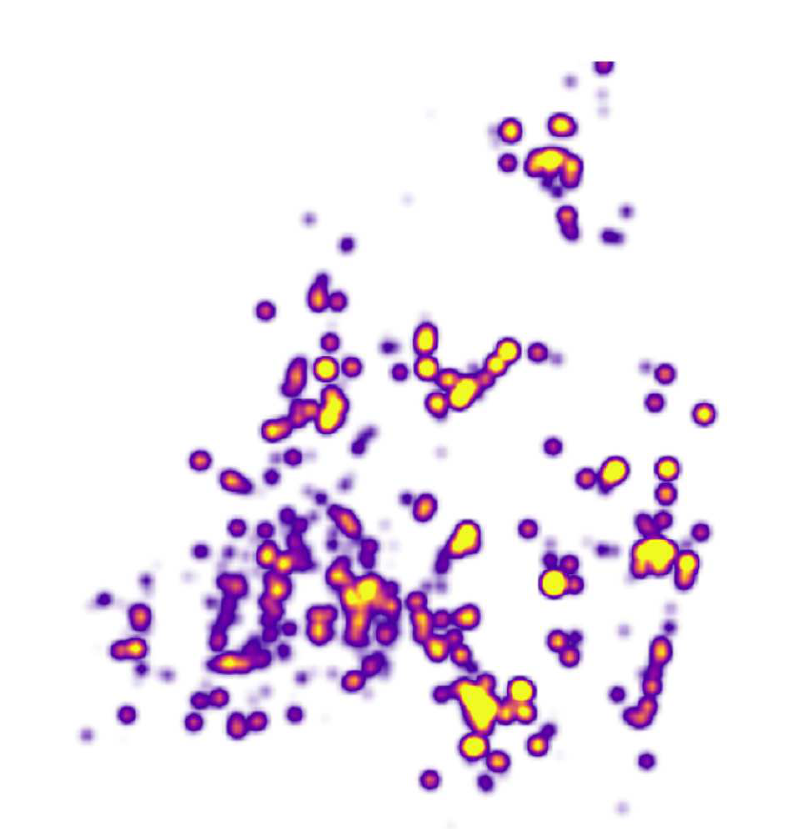}}
\subfigure[][Delivery Agents]{\includegraphics[width=1in]{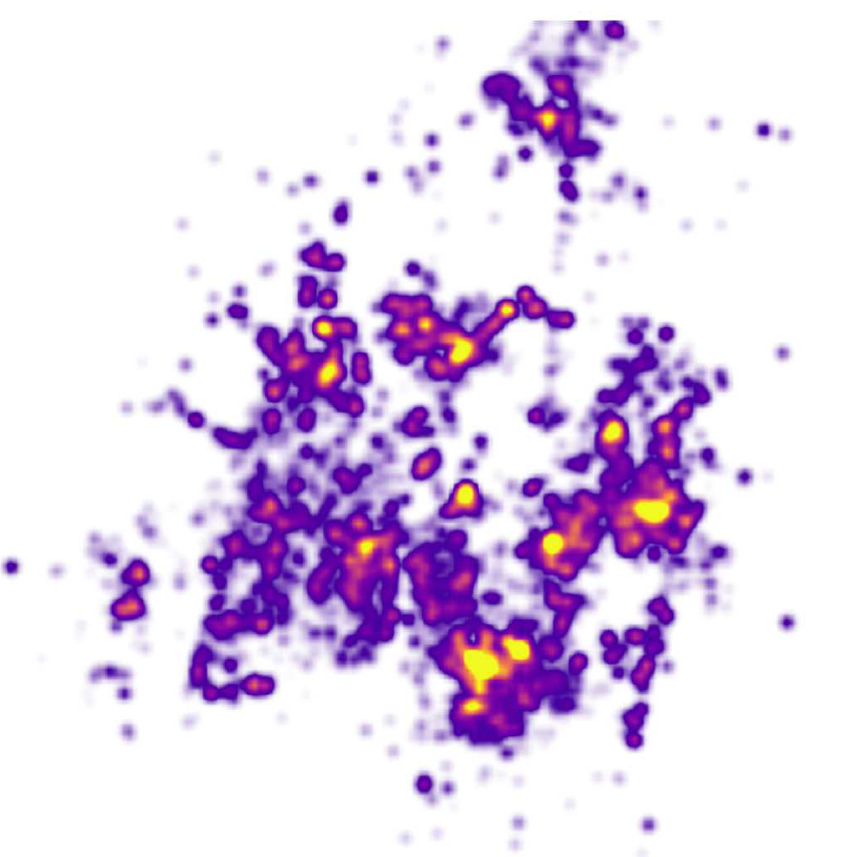}}
\caption{Heat map of order locations: the customers' locations, restaurants' locations and locations of the assigned delivery agents. In (a-c) we show this location data with respect to the orders serviced by the top $25\%$ agents and in (d-f) those serviced by bottom $25\%$ agents based on income.}
\label{fig:foodmatch_heatmap}
\end{figure}

\section{Ensuring Fairness in Food Delivery}
In this section, we define the concepts central to our work and formulate the problem of fair food delivery. 

\subsection{Fairness Notions}
Since food delivery platforms essentially {\it distribute income opportunities} among the delivery agents, the key question is: {\it what would constitute a fair/just distribution?} Fairness of distribution have been studied for a long time in Moral Philosophy, particularly in Distributive Justice~\cite{lamont2017distributive}. Next, we discuss few key principles from distributive justice and interpret them in the context of fair food delivery.

{\bf Strict Egalitarianism: } The underlying idea behind this fairness principle is that people are morally equal, and hence everyone should be treated equally~\cite{arneson2002egalitarianism}. In food delivery context, this would mean that every delivery agent should earn the same income from the platform. To implement this in practice, 
the platform should pull together all delivery fees collected and then distribute them equally among the agents. However, such schemes are practically untenable; more so due to the gig nature of delivery jobs.

{\bf Difference Principle: } In his seminal work on the theory of justice~\cite{rawls1971theory}, John Rawls defined a system to be just if those affected by the system agree to be subjected to it. Rawls permit a departure from equality only if it provides {\it `greatest benefit to the least advantaged members of society'}. In our context, this would translate into allocating the agent with the lowest income to a new order. However, this scheme would not consider the number of hours different agents work for the platform.


{\bf Proportional Equality: } Ronald Dworkin opposed the idea of complete equality and argued for eliminating inequality that happens by sheer luck, but allowing the impact of people's choice or hard work (known as `Luck Egalitarianism')~\cite{arneson2018dworkin}. In the food delivery context, agents' incomes should be proportional to their effort, i.e., the number of hours they are working. In this work, we consider this notion of {\it Proportional Equality}, 
and propose to ensure that an agent's income is proportional to the number of hours they work for the platform. 

\subsection{Background: The Food Delivery Problem}
\label{sec:inequality}
\noindent
Next, we formulate the problem of food delivery without any fairness consideration~\cite{foodmatch}. In general, the objective is to allocate orders to delivery agents such that the waiting time for customers is minimized. 

\begin{definition}[Road Network] 
\label{def:road-network}
\textit{A road network is a directed, edge-weighted graph $G = (V, E , \beta)$, where $V$ is the set of nodes representing regions, $E = \{(u, v): u,v \in V\}$ is the set of directed edges representing road segments connecting regions, and $\beta : (E,t) \mapsto \mathbb{R}^+$ maps each edge to a weight at time $t$. The edge weight at time $t$ denotes the expected time required to traverse the corresponding road at time $t$.}
\end{definition}
We use the notation $SP(u_i,u_{i+1},t)$ to denote the length of the \emph{shortest (quickest) path} from  $u_i$ to $u_{i+1}$ at time $t$.

\begin{figure}[t]
\centering
  \includegraphics[width=3.32in]{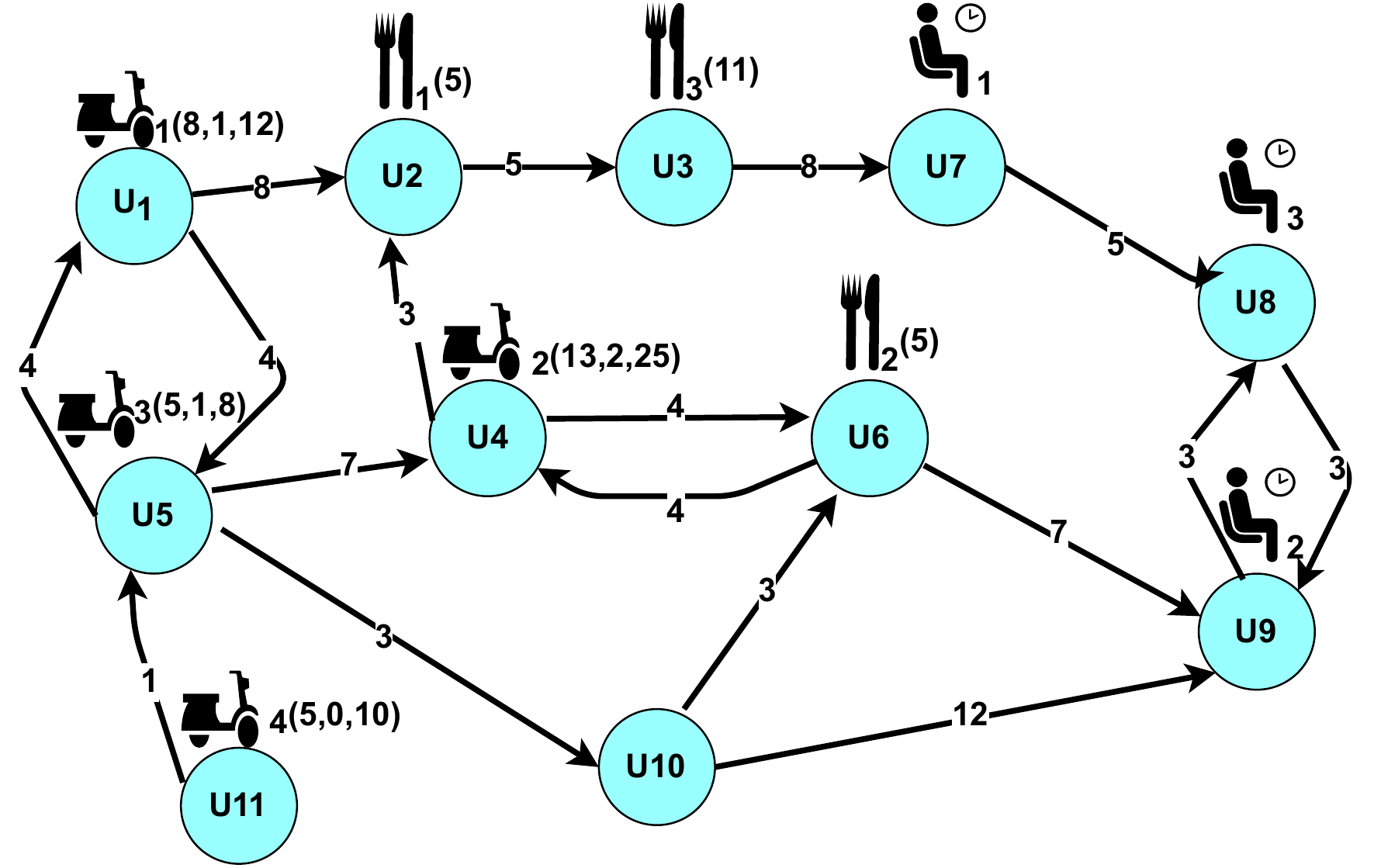}
  \caption{\bf A road network where the edge weights represent traversal times: vehicle icon represents delivery vehicle, cutlery icon represents restaurant and human icon represents customer location. Food preparation time at restaurant is in parentheses. Drive time, wait time and available time ($\driveTime,\waitTime, \availTime$) for the delivery agent is in parentheses next to the vehicle icons. The figure shows three order pick-up and drop-off points and three vehicles.}
  \label{fig:ex1}
\end{figure}
\begin{definition}[Food Order] 
\label{def:order}
\textit{A {\em food order} $o=\langle o^r, o^c, o^t,  o^p\rangle$ is characterized by four features.  $o^r\in V$ denotes the restaurant location (pick-up node), $o^c \in V$ is the customer location (drop-off node), $o^t$ is the time of request and $o^p$ is the (expected) food preparation time.}
\end{definition}


 Let $v$ be a food delivery vehicle. We use $O^v_t$ to denote the orders assigned to $v$. Furthermore, $loc(v,t)$ denotes the node that is closest to $v$ at time $t$. 
 We assume all vehicles have a maximum carrying capacity of $\maxorders$.  
Given a set of orders in $O^v_t$, a \textit{route plan} is a \textit{permutation} of $\{o_i^r, o_i^c: 1 \leq i \leq m\} \subset V$ such that for each $i$, $o_i^r$ appears before $o_i^c$ in the permutation. The \emph{length} of a route plan $RP=\{u_1,\cdots,u_m\}$  is $\sum_{i=1}^{m-1} SP(u_i,u_{i+1},t)$. The \emph{quickest} route plan is, therefore, the one with the smallest length. We assume that all vehicles always follow the quickest route plan, and hence any reference to a route plan is implicitly assumed to be the quickest one.

\begin{definition}[Order assignment] 
\label{def:assignment}
\textit{Given a set of orders $O$ and vehicles $\CV$, an order assignment function $A: O\rightarrow \CV$ assigns each order $o\in O$ to a vehicle $v\in\CV$. 
An order $o$ may be assigned to $v$ at time $t$ only if $|O^v_t| < \maxorders$.}
\end{definition}



Once order assignments are done, the \textit{first-mile} distance, $firstMile(o,v)$, of order $o$ is the distance from $v=A(o)$'s current location $loc(v,o^t)$ to the pick-up location $o^r$ in the route followed by $v$. Similarly, $lastMile(o,v)$ is the \textit{last-mile} distance from $o^r$ to drop-off location $o^c$.
\begin{example}
\label{ex:firstmile}
\textit{Let us consider Fig.~\ref{fig:ex1}. Suppose vehicle $v_1$, located at $u_1$, has been assigned to order $o_1$. $o_1$ needs to be picked up from the restaurant at $u_2$ and dropped-off at $u_7$. The quickest route for this task is $RP=\{u_1,u_2,u_3,u_7\}$.  
Thus, $firstMile(o_1,A)=8$ and $lastMile(o_1,A)=13$.} 
\end{example}

\begin{definition}[Expected Delivery Time]
\label{def:edt}
\textit{The expected delivery time of order $o$ when assigned to vehicle $v=A(o)$:}
\begin{alignat}{2}
\nonumber
EDT(o,v)&=\max\left\{ time(A(o))+firstMile\left(o,v\right),o^p\right\}\nonumber \\
\label{eq:edt}
&+lastMile\left(o,v\right)
\end{alignat}
\end{definition}
Here, $time(A(o,v))$ denotes the computation time taken by the assignment algorithm. To explain $EDT(o,v)$, the time to prepare food, and the time to assign a vehicle and reach the restaurant can progress in parallel. Thus, we take the maximum of these two components. If $o^p$ is larger, then the driver waits at the restaurant, which is loss of productive man hours. On the other hand, if $o^p$ is significantly smaller, then the food may get stale. 
\begin{example}
\label{ex:edt}
\textit{For simplicity, we assume $time(A(o))$ for all orders in Fig.~\ref{fig:ex1}. Continuing from Ex.~\ref{ex:firstmile}, $EDT(o_1,v_1)=\max\{8,5\}+13=21$. On the other hand, if $A(o_2)=v_2$ the quickest route plan is $\{u_4,u_6,u_9\})$, and thus $EDT(o_2,v_2)=\max\{4,5\}+7=12$.}
\end{example}
\begin{problem}[The food delivery problem (FDP)] 
\label{prb:online} 
\textit{Given a set of orders $O$ and vehicles $\CV$, if $\mathcal{A}$ is the set of all possible assignments of $O$ to $\CV$, find the assignment $A$ that minimizes the average expected delivery time.}
\begin{align}
\label{eq:objective}
    \arg\min_{A\in \mathcal{A}} \left\{\frac{1}{\lvert O \rvert}\sum_{\forall o\in O} EDT(o,A(o))\right\}
\end{align}
\end{problem}
At this juncture, we highlight two practical constraints.
\noindent
\textbf{(1)} In the real world, we work with a \textit{data stream} of orders and vehicles instead of sets. The typical strategy to circumvent this issue is to accumulate orders over a time window $\Delta$ and assign this order set to available vehicles~\cite{foodmatch,mdrp}. This process is then repeated over the data stream. 

\noindent
\textbf{(2)} A food delivery service provider guarantees a \textit{Service-level agreement (SLA}) of delivering the order within a stipulated time. This SLA is needed since food goes stale within a short time duration. Thus, it is desirable that the expected delivery time of \textit{all} orders is within the SLA threshold $\Omega$.
\subsection{Problem Formulation: Fair Food Delivery}
\label{sec:problem}
\noindent
Problem~\ref{prb:online} optimizes only the customer experience and does not incorporate the driver experience. As discussed earlier, a system would be fair if it equally distributes the \emph{time-normalized income} among all delivery agents. We therefore formalize the notion of fairness for delivery vehicle. 

\begin{definition}[Time-Normalized Vehicle Income]
Given any two time points $t_1 < t_2$ and a vehicle $v$, let $\availTime(v,t_1,t_2)$ be the time that $v$ was available in time interval $[t_1,t_2]$, and let $\driveTime(v,t_1,t_2)$ and $\waitTime(v,t_1,t_2)$ be the total time spent driving and waiting at restaurant respectively by vehicle $v$ in $[t_1,t_2]$. Then, if $\availTime(v,t_1,t_2) > 0$ we say that $v$'s {\em time-normalized income} in $[t_1,t_2]$ is defined as:
\begin{equation}
\label{eq:income}
\inc(v,t_1,t_2)=\frac{w_1\cdot \driveTime(v,t_1,t_2) + w_2\cdot \waitTime(v,t_1,t_2)}{\availTime(v,t_1,t_2)},
\end{equation}
where $w_1, w_2$ are payment parameters  decided by the food delivery company.
\label{def:income}
\end{definition}
Typically, $w_1>w_2$. Note that $\availTime(v,t_1,t_2) - \left(\driveTime(v,t_1,t_2) + \waitTime(v,t_1,t_2)\right)$ is the time during this interval when $v$ was either available, but had no orders assigned to it, i.e., it was idle.
\begin{example}
\label{ex:inc}
\textit{For Vehicle $v_2$ in Fig.~\ref{fig:ex1} $(\driveTime=13,\waitTime=2,\availTime=25)$ at $t=100$ so $\inc(v_2,0,100)= (w_1\cdot 13+w_2\cdot 2)/25 = 0.584$ where $w_1=1$ and $w_2=0.8$.}
\end{example}
\begin{problem}[\updatetext {Fair Income Distribution in Food Delivery}] 
\label{prb:fair} 
\textit{Given a set of orders $O$ beginning at time $0$ and ending at time $T_{m}$ and a set of available vehicles $\CV$, if $\mathcal{A}$ is the set of all possible assignments of $O$ to $\CV$, give an algorithm to find an assignment $A$ to minimize the income gap.}
\begin{align}
\label{eq:fair}
    \arg\min_{A\in \mathcal{A}} \left\{\max_{v\in \CV} \left\{\inc\left(v,0,T_{m}\right)\right\} - \min_{v\in \CV} \left\{\inc\left(v,0,T_{m}\right)\right\}\right\}
\end{align}
\end{problem}
\updatetext {Optimizing Problem~\ref{prb:fair} may lead to an increase in delivery time, and hence hamper customer's experience}. \updatetext {We therefore aim to minimize Problem~\ref{prb:fair} under \textit{bounded} increase in delivery time. Formally,}
\begin{problem}[\updatetext {The Fair Food Delivery Problem}]
\label{prob:multifair}
\updatetext{Find an assignment $A$ minimizing Problem~\ref{prb:fair} under the constraint: 
\begin{align} 
Dist(A(o),o) \leq \Gamma \times nearDist_{o} ,\forall o \in O
\end{align}
where $Dist(o,A(o)$ is the road network distance (shortest path) of order $o$ from its assigned vehicle $v=A(o)$, $nearDist_{o}$ is the distance to the nearest vehicle from order $o$ and $\Gamma>1$ is a threshold.}
\end{problem}

\subsection{Theoretical Characterization}
\begin{theorem}
\label{thm:fair}
There is no PTIME algorithm that can approximate Prob.~\ref{prb:fair} within any constant factor $c$, where $0 < c \leq 1$, unless $P = NP$.
\end{theorem}
\begin{proof} We prove it using a reduction from the \emph{subset sum} problem, which is a known NP-hard problem. Here, we show that if there exists a $c$-approximation algorithm for Prob.~\ref{prb:fair}, then we can separate the YES instance from the NO instances of subset sum in PTIME.

\begin{definition}[Subset Sum Problem]
Given a multi-set of items $I=\{a_1,\cdots,a_m\}$ and a target sum $B$, the subset sum is a decision problem that seeks to answer whether there exists a subset $I_1\subseteq I$, such that $\sum_{\forall a\in I_1} a=B$. We denote an instance of a subset sum problem with the notation $\mathcal{S}=\langle I,B\rangle$. $\mathcal{S}$ is a YES instance if the constraint is satisfied.
\end{definition}

Let us consider the instance $\mathcal{S}=\langle I,B\rangle$ where $\forall a_i\in I,\:a_i\geq 1$ and $B=\sum_{\forall a_i\in I}\frac{a_i}{2}$. Clearly the problem is still NP-hard. Given $\mathcal{S}$, we first construct a road network as follows. We first create a source node $S$. For each $a_i \in I$, we create a node $n_i$. From $S$ we have a directed edge to each $n_i$ and vice versa. We assume that there is an order corresponding to each $n_i$, such that the restaurant is located at $S$ with $0$ waiting time for all orders, and customer's location is $n_i$. The income from serving this request is assumed to be $a_i$. Thus, the maximum income that can be earned is $2\times B$.

We assume there are two vehicles and the maximum carrying capacity is $\maxorders=B$. Let $m^*$ be the income gap of the optimal solution of Prob.~\ref{prb:fair} on the constructed instance. We first show the following claim holds.
\begin{lemma}
If $\mathcal{S}$ is an YES instance of the subset sum problem, then $m^* = 0$.
\end{lemma}
\begin{proof}
If $\mathcal{S}$ is a YES instance, then there exists a partition of the set $I$ into $I_1$ and $I_2$ such that $\sum_{i \in I_1} a_i  = \sum_{i \in I_2} a_i = B$. Thus,  if we recommend the paths corresponding to elements of $I_1$ and $I_2$ to each of the two vehicles, then both earn an income of $B$, and hence, $m^*=0$. 
\end{proof}
\begin{lemma}
If $\mathcal{S}$ is a NO instance of the subset sum problem, then $m^*>0$.
\end{lemma}
\begin{proof}
If $\mathcal{S}$ is a NO instance, then clearly one of the two vehicles will earn less than $B$ and hence $m^*>0$. 
\end{proof}
This proves that Prob.~\ref{prb:fair} is NP-hard. Now, let us assume there exists a PTIME $c$-approximation., i.e., $m^* \leq m \leq cm^*$, where $m$ is the output of the approximation algorithm and $c\geq 1$. We know $m^*>0$ for NO instances. Since $m*\leq m$, $m> 0$. On the other hand, for YES instances, $m\leq cm^*\leq 0$ as $m^*=0$. Thus, the boundaries of $m$ for YES and NO instances will not overlap, and hence we will be able to separate those instances. This cannot be true unless $P=NP$. 
\end{proof}

\begin{corollary}
\updatetext{There exists no PTIME algorithm that can approximate Prob.~\ref{prob:multifair} within any constant factor $c$, where $0 < c \leq 1$, unless $P = NP$.}
\end{corollary}
\updatetext{\textsc{Proof.} Any instance of Prob.~\ref{prb:fair} reduces to an instance of Prob.~\ref{prob:multifair} for $\Gamma> diameter \times min\_distance$, where $diameter$ indicates the longest shortest path in the road network and $min\_distance$ is the minimum distance between any two points in the network.}

\subsection{Our Proposal: \fair}
\label{sec:ffm}
\noindent
We propose a heuristic algorithm \fair to solve Problem~\ref{prob:multifair}. It builds a \textit{weighted bipartite} graph with 
vehicles in one partition and \textit{clusters} of orders in the other. 
The weights of the edges are computed 
such that finding a \textit{ minimum weight matching} in this bipartite graph optimizes the criterion of  Prob.~\ref{prb:fair}, while also ensuring a good solution with respect to Prob.~\ref{prb:online}. We find the minimum weight matching by running the \textit{Kuhn-Munkres} algorithm~\cite{kuhn1955hungarian,munkres1957algorithms} on the graph. We now elaborate on the key steps.

\begin{definition}[Shortest Delivery Time~\cite{foodmatch}]
\label{def:sdt}
\textit{The {\em shortest delivery time} for an order $o$ is $SDT(o)= o^p + SP(o^r, o^c, o^t)$.} 
\end{definition}
$SDT(o)$ is a natural lower bound on EDT (Eq.~\ref{eq:edt}). 

\begin{example}
\label{ex:sdt}
\textit{For order $o_1$ in Fig.~\ref{fig:ex1}, food preparation time is $5$ and travel time from $restaurant_1$ at $u_3$ to $customer_1$ at $u_7$ is $(5+8)=13$, so  $SDT(o_1)=5+13=18$.}
\end{example}


\begin{definition}[Augmented Order Delivery Time]
\label{def:aodt}
Consider a vehicle $v$ at time $t$. Suppose at this time we add a cluster of orders $O$ to $v$'s route plan and $v$ is able to deliver all orders in this augmented route plan by time $t' > t$, on the assumption that no new orders are added to $v$'s route plan in time interval $(t,t')$. Then we define the  Augmented Order Delivery Time $AODT(O,v,t)$ as $t' - t$.
\end{definition}
\begin{example}
\label{ex:aodt}
\textit{If order $o_3$ in Fig.~\ref{fig:ex1} assigned to vehicle $v_1$ (at location $u_1$) at time  $t=100$ and $v_1$ was already assigned $o_1$, then $v_1$'s route plan will now be $RP = \{u_1,u_2,u_3,u_7,u_8\}$. $v_1$ will take $(8+5+8+5)=26$ time to deliver $o_3$ and it will reach $u_8$ at time $t'=100+26=126$ so $AODT(o_3,v_1,100) = t' - t = 126-100 =26$.}
\end{example}
\begin{definition}[Augmented Order Payment] 
\label{def:aop}
Consider a vehicle $v$ at time $t$ which may be currently idle or currently assigned some order. Suppose at this time we add a cluster of orders $O$ to $v$'s route plan (which may be empty if the vehicle is idle) and $v$ is able to deliver all orders in this augmented route plan at time $t'  > t$ such that the total driving time in the interval $[t,t']$ is $t_1$ and the total waiting time is $t_2$. Then, on the assumption that no new orders are added to $v$'s route plan in time interval $[t,t']$,  we define the {\em Augmented Order Payment} $AOP(O,v,t)$ as $w_1\cdot t_1 + w_2 \cdot t_2$. 
\end{definition}
\begin{example}
\label{ex:aop}
\textit{Suppose at $t=100$ order $o_1$ in Fig.~\ref{fig:ex1} was already assigned to $v_2$ located at $u_4$. At this point, $v_2$'s route plan is $\{u_4,u_2,u_3,u_7\}$. Suppose at this time, $o_3$ is assigned to $v_2$. Now, $v_2$ will follow route plan $RP = \{u_4,u_2,u_3,u_7,u_8\}$. $v_2$ will travel 3 time units to reach $restaurant_1$ located at $u_2$, then it will wait $2$ time units till the order is ready. At $t = 105$, it will leave $u_2$ and travel 5 time units to $u_3$ where it will wait $1$ time unit for $o_3$ to be prepared. Finally, $v_2$ will set out from $u_3$ at time $t = 111$ and, after delivering $o_1$ and $u_7$ at time $t = 119$, finish its augmented route plan by delivering $o_3$ at $u_8$ at time $t = 124$. The total time spent is therefore $124 - 100 = 24$, of which $2+1 = 3$ time units were spent waiting. So $AOP(o_3,v_2,100) = 21 \cdot w_1 + 3 \cdot w_2$ = 23.4 where $w_1=1$ and $w_2=0.8$. Recall, $AODT(o_3,v_2,100) = 24$ (Ex.~\ref{ex:aodt}}).
\end{example}

As discussed in our problem formulation (Prob 1), we partition the data stream of orders into windows of length $\Delta$, and  allocate all orders that arrived within this window. This process is then repeated for each of the subsequent windows. 
\begin{definition}[Next-slot Normalized Income]
\label{def:ns_inc}
At time $\ell \Delta$, i.e. at the beginning of the $\ell +1$st window, consider a vehicle $v$ that is available. In addition, consider an unassigned order $o$. We define the next-slot normalized income of $v$ if it is assigned $o$ as 
\begin{equation}
\label{eq:inc'}
\ninc(v,o,\ell) = \frac{\inc(v,\ell \Delta)\cdot \availTime(0,\ell \Delta) + AOP(o,v,\ell \Delta)}{\ell \Delta + AODT(o,v,\ell \Delta)}.
\end{equation}
\end{definition}
\begin{example}
\label{ex:ns_inc}
\textit{Continuing from Ex.~\ref{ex:aop} for Vehicle $v_2$ in Fig.~\ref{fig:ex1} $(dT=13,wT=2,aT=25)$ at $t=100$, $AODT(o_3,v_2,100) = 24$,  $AOP(o_3,v_2,100) = 23.4$ and $\inc(v_2,0,100)= 0.584$ (see Ex.~\ref{ex:inc}). So
$\ninc(v_2,o_3,100) = \frac{\inc(v_2,0,100)\cdot 25 + AOP(o_3,v_2,100)}{25 + AODT(o,v,100)} = (0.584\cdot 25+ 23.4)/(25+24)=0.76$.
}
\end{example}
\subsubsection{Creating a weighted bipartite graph:}
\label{sec:bipartite}
  At time $\ell \Delta$, let us assume that $\CV_\ell$ is the set of available vehicles and $O_\ell$ is the set of unallocated orders. We create a weight bipartite graph $(U_1, U_2,E)$ where $U_1 = \CV_\ell$, i.e., one side of the partition is the set of available vehicles and $U_2$ is a \textit{cluster} of orders.

\noindent
\textbf{Clustering orders:} If $|O_\ell| \leq f\cdot |\CV_\ell|$ where $f$ is a parameter chosen in $(0,1)$, we set $U_2 = O_\ell$. Otherwise we perform a clustering on $O_\ell$ using Ward's method, i.e., we successively coalesce those two clusters whose being delivered by a single vehicle leads to the least increase in the extra delivery time. The cluster size is not allowed to cross \maxorders. We stop either when the number of clusters falls to $f \cdot |\CV_\ell|$ or when the increase in expected delivery time due to clustering crosses a threshold $\eta$. More simply, we stop when either further clustering significantly compromises customer experience or there are too few clusters to fairly allocate to all available vehicles. 
We denote the final clusters as $O'_\ell$.

\noindent
\textbf{Edge weights in bipartite graph:}
Each edge of the bipartite graph is of the form $(v,O)$ where $O \subset O_\ell$ is either a single order (a singleton set) or a cluster of orders. 
\updatetext{If vehicle $v$ is at a distance not exceeding $\Gamma\times nearDist_{O}$ (Recall Prob.~\ref{prob:multifair})}, we set the weight as follows:
\begin{equation}
\label{eq:weight_bipart}
w(v,O) = \textit{\ninc}(v,O,\ell) - \min_{v\in \CV} \textit{\inc}(v,\ell \Delta).
\end{equation}
To identify all vehicles whose distance not exceeding $\Gamma\times nearDist_{O}$, in an efficient manner, we perform \textit{best-first search} on the road network graph (See Alg.~\ref{algo:fg}). Specifically, we start from the restaurant locations of each order in $O$ and visit all the nearby vehicles in the best first search order till the distance from source restaurant exceeds $\Gamma\times nearDist_{O}$. All vehicles beyond this boundary are assigned edge weight $\approx \infty$.

Finally, we run \textit{Kuhn-Munkres} on the bipartite graph to obtain the allocation for the $\ell+1$st window.

\subsection{Algorithm}

\begin{algorithm}[h!]
    \caption{Bipartite graph construction }\label{algo:fg}
    {\scriptsize
    \begin{flushleft}
        \textbf{Input:} Available vehicles $\CV_\ell$, order batch $O'_\ell$, parameter $\Gamma$, current time $t$, road network $G(V,E,\beta)$\\
        \textbf{Output:} Bipartite Graph
    \end{flushleft}

    \begin{algorithmic}[1]
    \State Initialize bipartite graph $B$ with node sets $\CV_\ell$ and $O'_\ell$, and empty edge set $E_b$.
    \Foreach {$o'\in O'_\ell$}
	\State $source \leftarrow loc(o',t)$
	\State $PQ\leftarrow$ Empty Priority Queue
	\State $PQ.insert(\langle source,0 \rangle)$
	\State Initialize $\forall u\in V,\:visited(u)\leftarrow false$
	\State $\CV_{o'}\leftarrow \emptyset$
	\State Set $found\_nearest\_v=false$
	\State Set $nearDist_{o'}=0$
    \While {$PQ.empty()=false$}
	\State $\langle u, \delta \rangle\leftarrow PQ.pop()$
	    \If {($found\_nearest\_v=true \And \delta>=\Gamma\times nearDist_{o'}$)} 
	    \State{Break}
	    \EndIf
        \IIf {$visited(u)=true$} \textbf{continue}
	\State $visited(u)\leftarrow true$
	\State $I(u)\leftarrow\{v\in\CV_\ell\:\mid\:loc(v,t)=u \}$
	    \If {($I(u)\neq NULL \And found\_nearest\_v=false$)}
	    \State Set $nearDist_{o'}=\delta$
	    \State Set $found\_nearest\_v=true$
	    \EndIf
	 
	\Foreach {$v\in I(u)$} 
		\State Add edge from $v$ to $o'$ in $E_b$ with edge weight given by Eq.~\ref{eq:weight_bipart}.
	\EndForeach
	\State $\CV_{o'}\leftarrow \CV_{o'}\cup I(u)$
	\State $N(u)\leftarrow (u'\:\mid\: (u,u')\in E,\:visited(u)=False,\text)$
	\State $\forall u'\in N(u),\:PQ.insert(\langle u',\delta+\beta(e=(u,u'),t)\rangle)$
    \EndWhile
	\State $\forall v\in\CV_\ell\setminus \CV_{o'}$, add edge from $v$ to $o'$ in $E_b$ with $\Omega$ edge weight.
    \EndForeach
    \State \Return bipartite graph $B(\CV_{\ell},O'_\ell,E_b)$
    \end{algorithmic}}
\end{algorithm}

Algorithm ~\ref{algo:fg} represents the pseudocode for bipartite graph construction using best first search approach. For a given batch of orders $o'\in O'_\ell$ we extract its restaurant location of first unpicked order $source \leftarrow loc(o',t)$ (line 3) and from source we initiate best first search (lines 10-22). A priority queue $PQ$ is used to store candidate nodes to be visited and is initialized with source (line 4-5).
$PQ$ stores a tuple  $\langle u, \delta_u \rangle$ where $\delta_u$ is $SP(source, u, t)$, the shortest path from $source$ to $u$, and retrieves nodes in ascending of $\delta_u$. Once top node from $PQ$ is popped (line 11), we proceed further if it's distance does not exceed $\Gamma\times nearDist_{o'}$(line 12-13) and the node is not yet visited (line 14). Next, we mark $u$ as visited (line 15). We take all vehicles that are at $u$ (line 16). We set nearest vehicle distance $nearDist_{o'}$ to $\delta$ if nearest vehicle is not yet found (lines 17-19). We add an edge from $o'$ to all vehicles available at $u$ in edge set $E_b$ with edge weight given by Eq.~\ref{eq:weight_bipart} (lines 20-22). After this, we insert all neighbors $u'$ of $u$ to $PQ$ with $ \delta' = \delta+\beta(e=(u,u'),t)$ where $\beta(e=(u,u'),t)$ is edge weight of $e=(u,u')$, i.e, average time taken in road network at time $t$ (lines 24-25). Once $PQ$ is empty, we add an edge from $o'$ to all remaining vehicles with weight $\Omega$ ($\approx \infty$)(line 26). This whole process is repeated for all batches of orders $o'\in O'_\ell$ (line 2).

 \noindent
 \begin{theorem}
 The time complexity of the proposed algorithm is  $\mathcal{O}(m\cdot n(q+ \max(m,n)))$, where $m=|\CV_\ell|$, $n=|O'_\ell|$, $\mathcal{O}(q)$ is the time taken for shortest path computation.
 \end{theorem}
 \begin{proof} 
To create an optimal route plan, we need all permutations of at most $2 \maxorders$ locations with the constraint that the customer location for an order must come after the restaurant for that order. This is the same as the ways of arranging $\maxorders$ pairs of parentheses with the twist that in our cases the parentheses are distinguishable. So this number is $G(\maxorders) = (2\maxorders)!/((\maxorders + 1)(\maxorders)!)$. For each permutation we need to find the shortest path between subsequent nodes in the order. Hence the time taken to find the route plan is $\mathcal{O}(G(\maxorders) \cdot \maxorders \cdot q)$ where $\mathcal{O}(q)$ is the time taken for shortest path computation.  We compute $w(v,o')$ for each vehicle-order pair which takes a total of  $\mathcal{O}\left((G(\maxorders) \cdot \maxorders \cdot q\cdot m\cdot n\right)$ time where $m=|\CV_\ell|$ and $n=|O'_\ell|$.
 Let $k_{\top}=\max(n,m)$ and $k_{\bot}=\min(n,m)$. Kuhn-Munkres algorithm take  $\mathcal{O}(k_{\top}^2 k_{\bot})$ time to compute a maximum weighted matching. Hence an overall complexity is $\mathcal{O}((G(\maxorders) \cdot \maxorders \cdot q\cdot m\cdot n+ k_{\top}^2 k_{\bot})=\mathcal{O}(G(\maxorders) \cdot \maxorders \cdot q\cdot m\cdot n+ m\cdot n\cdot \max(m,n))$.
 
 Typically, $\maxorders$ is $2$ or $3$, and hence $\maxorders \ll m,n,q$. Thus, we drop this term, which reduces the complexity to $\mathcal{O}(m\cdot n(q+ \max(m,n)))$.
 
 \end{proof}

\section{Experimental Evaluation}
\label{sec:experiments}
\noindent
In this section, we benchmark \fair and establish:
\begin{itemize}
    \item \textbf{Fairness:} \fair imparts more than $10X$ improvement in  fairness over \fmplus and baseline approaches adopted from the cab service industry.
    \item \textbf{Cost of fairness:} \fair maintains a comparable delivery time as that of \fmplus~\cite{foodmatch}.
    \item \textbf{Scalability:} \fair is scalable to real-world workloads in large metropolitan cities.
\end{itemize}

\subsection{Evaluation Framework} 
\label{sec:setup}
\noindent
All implementations are in C\texttt{++}. Our experiments are performed on a machine with Intel(R) Xeon(R) CPU @ 2.10GHz with 252GB RAM on Ubuntu 18.04.3 LTS. Our codebase is available at 
{\color{black}
\url{https://github.com/idea-iitd/fairfoody.git}.} 

\noindent
\textbf{Baselines:}\\
$\bullet$ \textbf{\fmplus:} \fmplus squarely focuses on minimizing the delivery time. Hence, it is agnostic to driver incomes. A comparison with \fmplus reveals \textbf{(1)} the unfairness in the system when we optimize only delivery time, and \textbf{(2)} the impact on customer experience in terms of delivery time if fairness is included as an additional objective in the optimization function.\\
$\bullet $\textbf{\abhi:} \abhi is designed to ensure two-sided fairness in the cab-hailing industry. The two sides correspond to cab drivers and customers. Towards that end, \abhi optimizes a weighted combination of driver's income and the wait-time faced by customers (delivery time in the context of food delivery). We use $\lambda$ to denote the weightage given to drivers income; $1-\lambda$ corresponds to weightage of waiting time for customers. Although there are similarities between the cab and food delivery industry, there are several subtle differences that necessitates the need for a specialized algorithm for food delivery. For example, sending the nearest cab driver minimizes the waiting time of a customer, whereas that is not the case in food delivery due to the intermediate operation of picking up food from the restaurant. This comparison allows us to precisely understand the impact of a food-delivery specific fairness algorithm.\\
We use \textit{hierarchical hub labeling}~\cite{hhl} to index shortest paths queries in all benchmarked algorithms. 
 
\noindent
\textbf{Metrics:} 
The performance is quantified through:

\textbf{$\bullet$ Gini coefficient:} 
The Gini coefficient is the ratio of the area that lies between the line of equality and the Lorenz curve over the total area under the line of equality \cite{gini_coef}.
 Mathematically,
 \begin{equation}
 \label{eq:gini}
Gini=\frac{\sum_{i=1}^{n} \sum_{j=1}^{n}\left|x_{i}-x_{j}\right|}{2 n \sum_{j=1}^{n} x_{j}}
\end{equation}
where $x_i$ is the income per hour of driver $i$ and $n$ is the number of drivers. A lower Gini indicates fairer distribution.

 
 \textbf{$\bullet$ DTPO:} DTPO measures the average delivery time per order. DTPO allows us to quantify the cost of fairness.

\textbf{$\bullet$ Percentage of SLA violations (SLA-V):}  We measure the percentage of orders not delivered within the promised time limit (Prob.~\ref{prb:online}). The time limit is set to $45$ minutes.

\begin{table}[b]
    \centering
    \scalebox{0.9}{
    \begin{tabular}{|l|l|l|l|l|l|}
    \hline
       \textbf{City} &  \textbf{Algorithm} & \textbf{Gini} & \textbf{Income} & \textbf{DTPO} & \textbf{SLA-V} \\
       \textbf{} &  \textbf{} & \textbf{} & \textbf{Gap} & \textbf{} & \textbf{(\%)} \\ \hline
        \multirow{4}{*}{\textbf{A}}&\fair & \textbf{0.035} & \textbf{20.5} & 15.4 & 0.33 \\ 
        &\abhi, $\lambda=1$ & 0.32 & 58 & 17.3 & 0.37 \\ 
        &\abhi, $\lambda=0$  & 0.526 & 55.9 & \textbf{15.2} & 0.33 \\ 
        &\fmplus & 0.518 & 59.2 & \textbf{15.2} & \textbf{0.32} \\ \hline
       \multirow{4}{*}{\textbf{B}} &\fair  & \textbf{0.047} & \textbf{30.4} & 15.5 & \textbf{0.22} \\ 
        &\abhi $\lambda=1$  & 0.316 & 59.9 & 15.5 & 23.12 \\ 
        &\abhi, $\lambda=0$ & 0.471 & 59.4 & \textbf{14.5} & 19.62 \\ 
        &\fmplus & 0.512 & 59.3 & 15.4& 0.23 \\ \hline
       \multirow{4}{*}{\textbf{C}} &\fair & \textbf{0.035} & \textbf{33.5} & 16.2 & 0.33 \\ 
        &\abhi, $\lambda=1$ & 0.323 & 58.3 & 16.9 & 12.05 \\ 
        &\abhi, $\lambda=0$  & 0.513 & 56.1 & \textbf{15.8} & 6.63 \\ 
        &\fmplus  & 0.562 & 59.1 & 16.0 & \textbf{0.32} \\ \hline
    \end{tabular}}
    \caption{Performance of \fair, \fmplus, and \abhi across various metrics. The unit of DTPO is in minutes. The best performance in each metric is highlighted in bold.}
    \label{tab:metric}
\end{table}
\textbf{$\bullet$ Spatial distribution distance ($\psi$):} 
we examine the spatial distribution of top-$25$\% and bottom-$25$\% drivers in terms of income in the form of a \textit{heatmap}. We next formulate a metric to quantify the difference between these distributions. 
Specifically, we partition a city into a grid and compute three normalized distributions over the cells of the grid for the top-$25$\% and bottom-$25$\% drivers: \textbf{(1) $\CL$:} their locations when an order was allocated to them, \textbf{(2) $\CR$:} the restaurant location of allocated orders, and \textbf{(3) $\cc$:} the customer drop-off locations. If there are $n$ grid cells, then they give rise to an $n$-dimensional vector with non-negative entries that sum to 1 for each property. We denote by $\alpha_P$ and $\beta_P$ the distributions of top-$25$\%  and bottom-$25$\% drivers respectively on property $P$. Since $\alpha_P$ and $\beta_P$ vectors have the property of probability distributions we compute the distance between them using the\textit{ Total Variation Distance} (c.f.~\cite{Levin-BOOK:2017}), which is guaranteed to be between $0$ and $1$, i.e.,  
\begin{equation}
 \label{eq:HeatmapCoeff}
\psi_P=\frac{1}{2}{\sum_{i=1}^{n}\left\lvert\alpha_P[i] - \beta_P[i]\right\rvert}
\end{equation}
\noindent
\textbf{Operational Constraints: } Several constraints are required while operating food-delivery service. We use the same constraints adopted by\cite{foodmatch}. An order is \emph{rejected} if it remains unallocated for $45$ minutes.
The rejection penalty $\Omega$ (Recall Problem~\ref{prb:online}) is set to $7200$ seconds since most orders are delivered in motorbikes.\\ 
\noindent
\textbf{Simulation Framework} Our dataset contains the exact position of all orders, restaurant and delivery agents. The simulation environment to evaluate the impact of each allocation algorithm, therefore, only involves the order allocation mechanism. Order allocation is a deterministic procedure and hence we neither need to repeat the experiment multiple times, nor report the variance. Our reported results measure the performance on the various described metrics across all 6 days of data.\\
\textbf{Estimation of road network speeds:} We divide $24$-hour period into $24$ one hour slots and then for each slot we compute the weight of each road network edge as the average travel time across all of delivery vehicles in the corresponding road in that slot.\\
\noindent
\textbf{Parameters: }The default size of accumulation window $\Delta$ is $3$ minutes. Clustering parameter $f$, and payment weights $w_1$ and $w_2$ (Def.~\ref{def:income}) are set to $0.8$, $1.0$ and $0.8$ respectively as their default values. 
\begin{table}[t]
    \centering
    \scalebox{0.9}{
    \begin{tabular}{|l|l|l|l|l|l|}
    \hline
        \textbf{City} & \textbf{Algorithm} & \textbf{25 \permil} & \textbf{50 \permil} & \textbf{75 \permil} & \textbf{95 \permil} \\\hline
        \multirow{3}{*}{\textbf{A}}&\fair & 38 & 51.5 & 70 & 88 \\
        & \fmplus & 13 & 29 & 65 & 205.3 \\
        & \abhi $\lambda=1$ & 20 & 44 & 80 & 142 \\ \hline
        \multirow{3}{*}{\textbf{B}}& \fair & 31 & 61 & 92 & 115 \\
        & \fmplus & 12 & 33 & 75 & 264 \\
        & \abhi $\lambda=1$ & 23 & 53 & 90 & 156 \\ \hline
        \multirow{3}{*}{\textbf{C}}& \fair & 44 & 64 & 79 & 95 \\
        & \fmplus & 13 & 31 & 74 & 250 \\ 
        & \abhi $\lambda=1$ & 28 & 52 & 87 & 154 \\ \hline
    \end{tabular}}
    \caption{Number of orders across all six days per driver at various percentiles(\permil).}
     \label{tab:percentile}
\end{table}
\subsection{Comparison with \fmplus and \abhi}
\noindent
Table~\ref{tab:metric} presents the performance of various algorithms across all three cities. In Fig.~\ref{fig:foodmatch_gini_norm}, we observed that \fmplus induces significant income disparity among drivers. This is reflected in the high Gini of \fmplus across all cities in Table~\ref{tab:metric}. In contrast, \fair reduces Gini more than 10 times across three cities. This reduction in Gini, however, does not come at the cost delivery time or SLA violations (SLA-V). Specifically, there is minimal increase in DTPO and SLA-V. Similar to Gini, \fair is also significantly better in Income Gap. Overall, this shows that it is possible to ensure fairness without compromising on the customer experience. 

We also compare with \abhi at $\lambda=1$ and $\lambda=0$. We choose these two $\lambda$ values since they represent the two extremes; at $\lambda=1$, \abhi optimizes only the driver income gap, whereas $\lambda=0$ minimizes only the delivery time. Thus, $\lambda=1$ represents the best possible Gini by \abhi. We observe that even in this scenario, \fair is $9$ times better on average. At $\lambda=0$, although \abhi achieves low delivery times, it fails to satisfy SLA across a large portion of orders.

In Table~\ref{tab:percentile}, we further examine the number of orders delivered by drivers across various percentiles based on normalized income (Def.~\ref{def:income}). It is clear from the data that \fair achieves the most equitable distribution.
\begin{table}[b]
    \centering
    \scalebox{0.9}{
    \begin{tabular}{|l|l|l|l|l|l|l|}
    \hline
        \textbf{City} & \textbf{Locations} &\textbf{\fmplusone} & \textbf{\abhi} & \textbf{\abhi} & \textbf{\fairone} \\
        \textbf{} & \textbf{of the} &\textbf{\fmplustwo} & \textbf{$\lambda=0$} & \textbf{$\lambda=1$} & \textbf{\fairtwo} \\
       \hline
        \multirow{3}{*}{\textbf{A}}& Vehicle & 0.543 & 0.500 & 0.293 & \textbf{0.231} \\
        & Customer & 0.490 & 0.437 & 0.268 & \textbf{0.235} \\ 
        & Restaurant & 0.469 & 0.442 &  \textbf{0.169} & 0.278 \\\hline
        \multirow{3}{*}{\textbf{B}}& Vehicle & 0.394 & 0.394 &  \textbf{0.169} & 0.191 \\
        & Customer & 0.348 & 0.312 &  \textbf{0.152} & 0.182 \\ 
        & Restaurant & 0.354 & 0.342 &  \textbf{0.135} & 0.174 \\\hline
        \multirow{3}{*}{\textbf{C}}& Vehicle & 0.386 & 0.403 &  0.243 & \textbf{0.193} \\
        & Customer & 0.318 & 0.310 &  0.217 & \textbf{0.185} \\
        & Restaurant & 0.318 & 0.382 &  \textbf{0.206} & 0.239 \\ \hline
    \end{tabular}}
    \caption{Comparison of spatial distribution $\psi_P$.}
    \label{tab:heatmap}
    
\end{table}
\subsection{Impact on Spatial Distribution}
\noindent
In Fig.~\ref{fig:foodmatch_heatmap}, we showed that spatial distribution is a key driver of payment inequality. Now we will show that {\em \fair equalizes the spatial distributions across the pay range}. Fig.~\ref{fig:fairfoody_heatmap} studies this question in City C. As visible, the heatmaps of the top-$25\%$ and bottom-$25\%$ are much more similar when compared to Fig.~\ref{fig:foodmatch_heatmap}. This indicates that disparity in spatial distribution is correlated to income disparity.

To quantify this observation and examine whether the pattern holds across all cities, we compute $\psi_P$ across all combinations of cities and properties (Recall Eq.~\ref{eq:HeatmapCoeff}). Table~\ref{tab:heatmap}
 presents the results. Two key observations emerge from this experiments. First, both for \fair and \abhi, $\lambda=1$, the $\psi$ is lower across all cities and properties. This indicates that spatial distribution distance and Gini (as well as Income Gap) are indeed correlated. However, just minimizing spatial distribution distance is not enough in minimizing Gini. Specifically, although \abhi, $\lambda=1$ has a lower spatial distance than \fair in City B, its Gini is significantly higher. (Table~\ref{tab:metric}).
\begin{figure}[t]
\centering  
\subfigure[][Customers]{\includegraphics[width=1.3in]{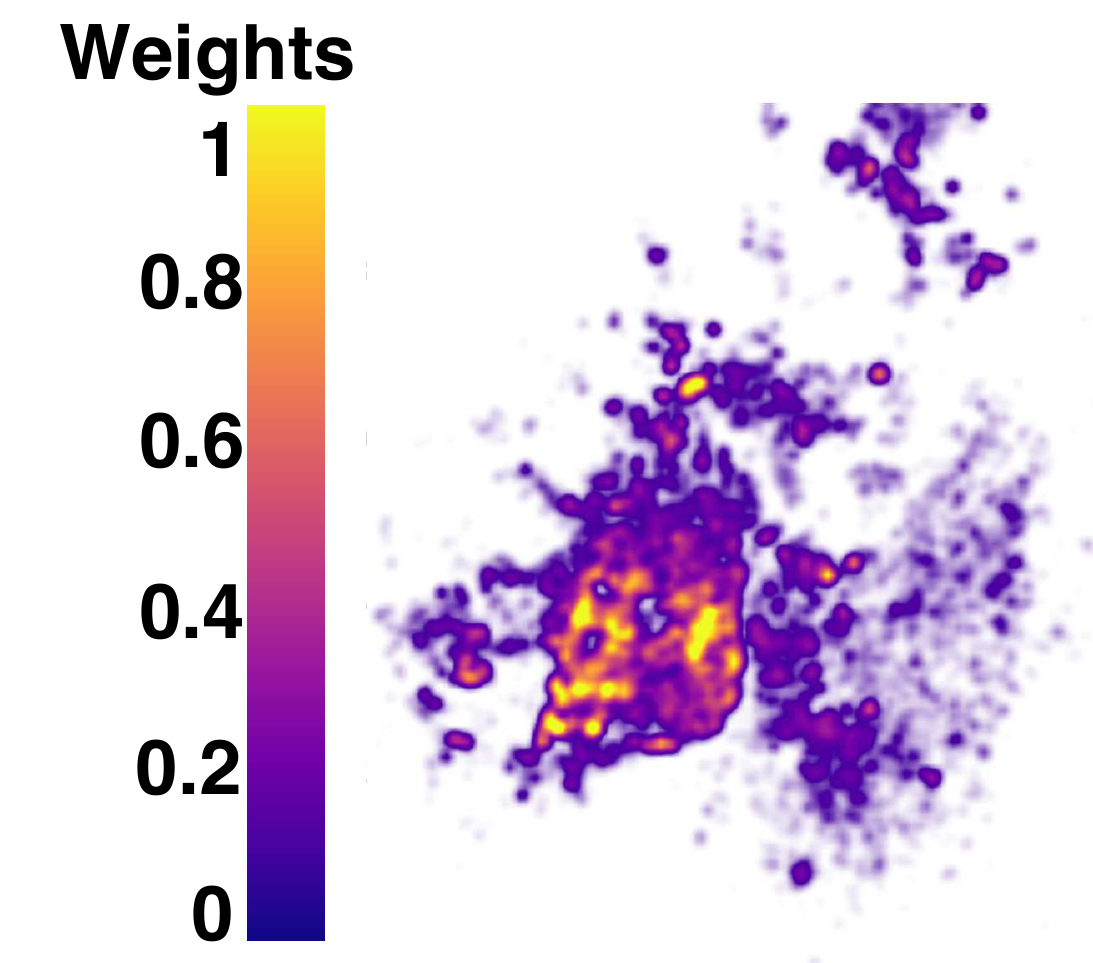}}\hfill
\subfigure[][Restaurants]{\includegraphics[width=1in]{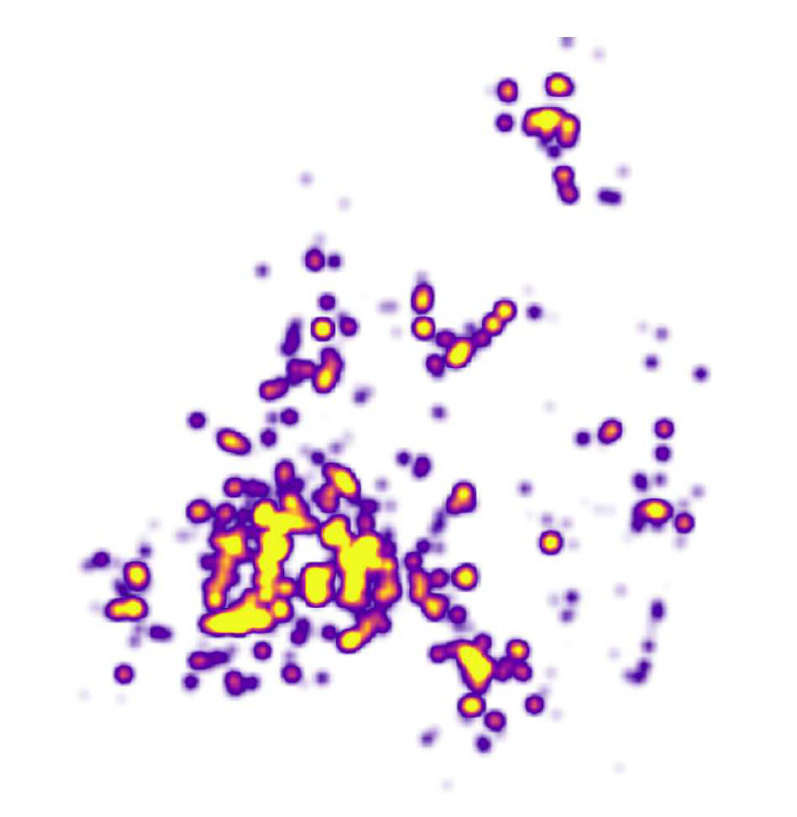}}\hfill
\subfigure[][Delivery Agents]{\includegraphics[width=1in]{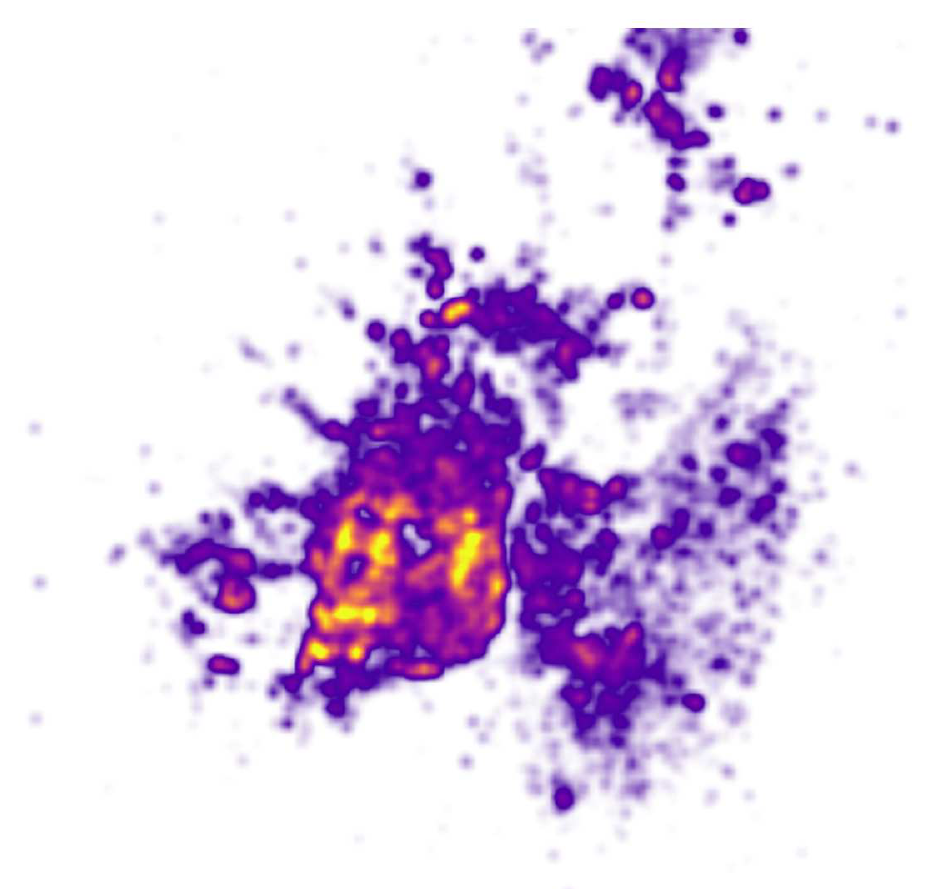}}\\
\subfigure[][Customers]{\includegraphics[width=1in]{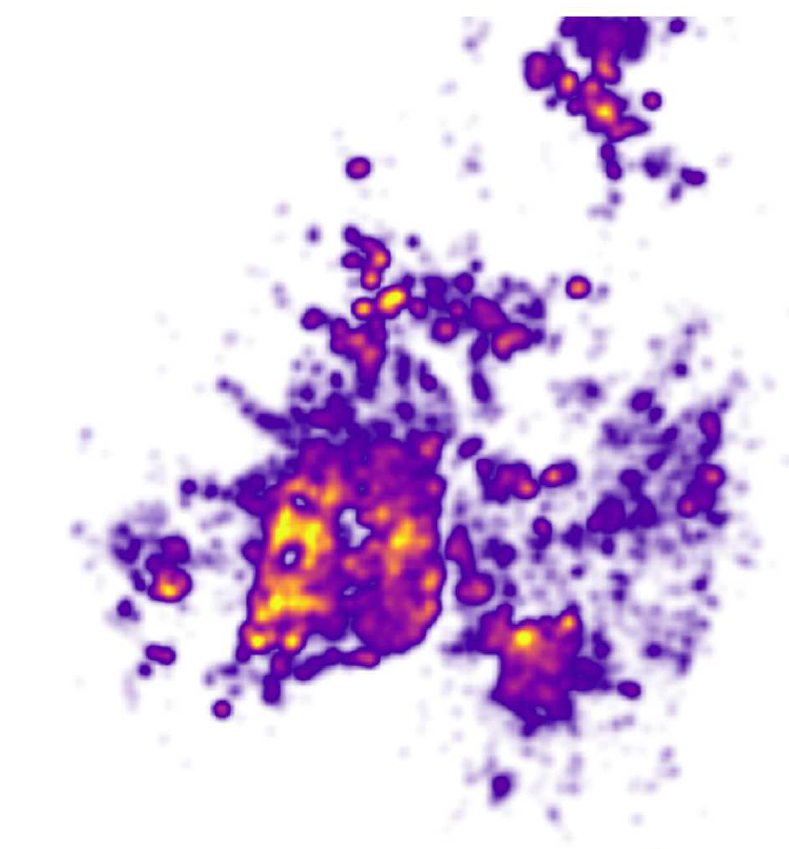}}\hfill
\subfigure[][Restaurants]{\includegraphics[width=1in]{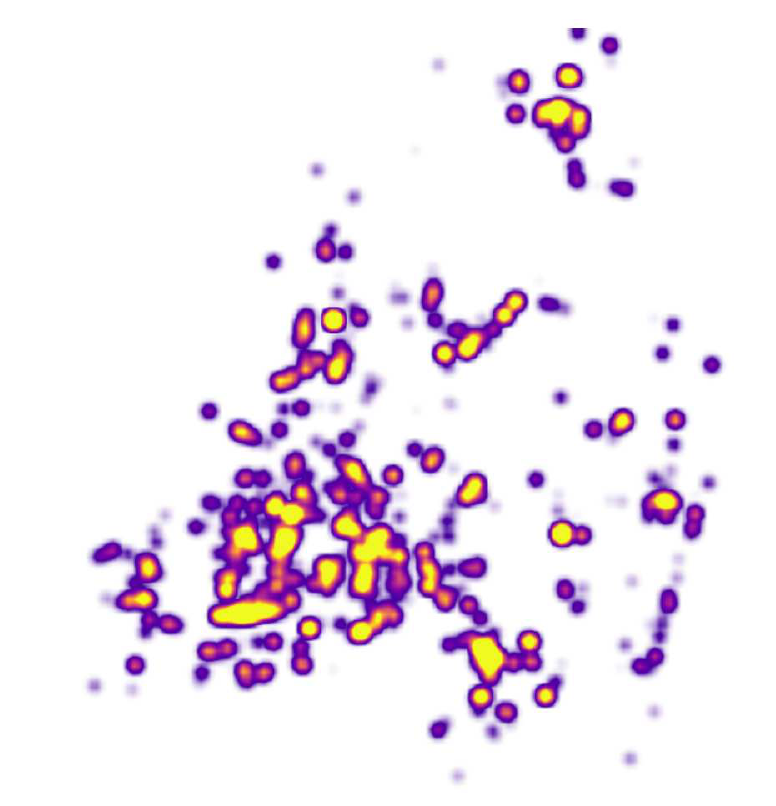}}\hfill
\subfigure[][Delivery Agents]{\includegraphics[width=1in]{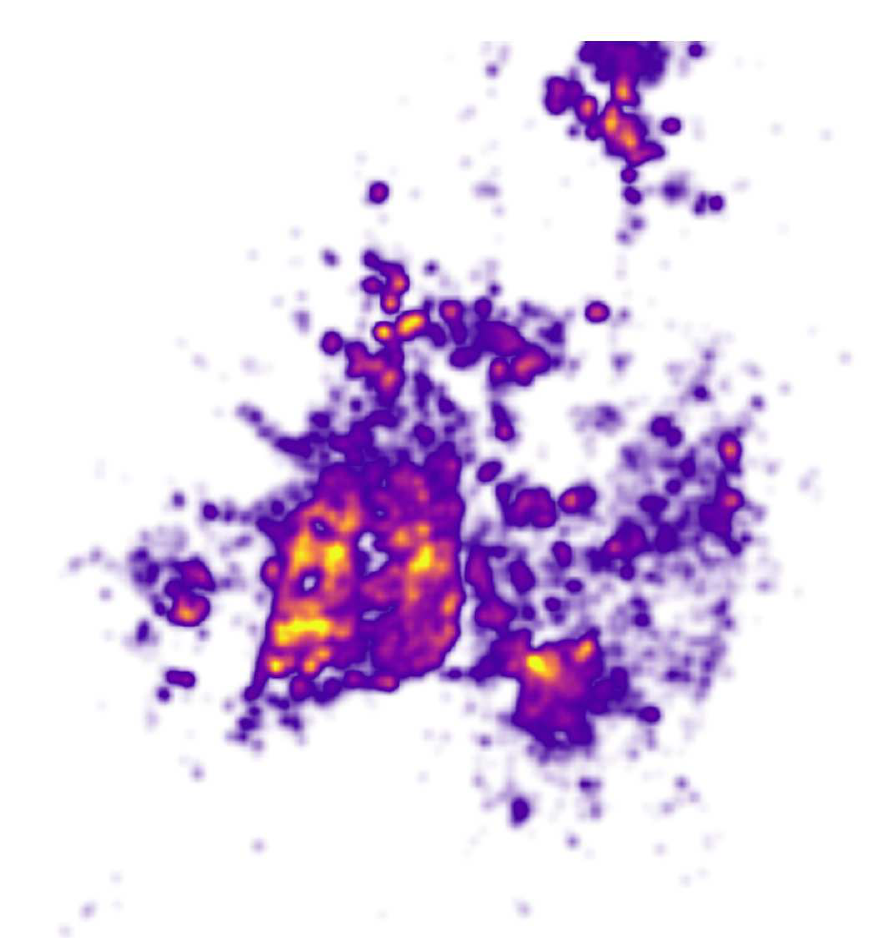}}
\caption{Heat map of order locations after applying \fair. 
In (a-c) we show this location data with respect to the orders serviced by the top $25\%$ agents and in (d-f) those serviced by bottom $25\%$ agents based on income.} 
\label{fig:fairfoody_heatmap}
\end{figure}
\begin{figure*}[t]
\centering  
\subfigure{\label{fig:clustering_variation_gini}\includegraphics[width=0.24\linewidth]{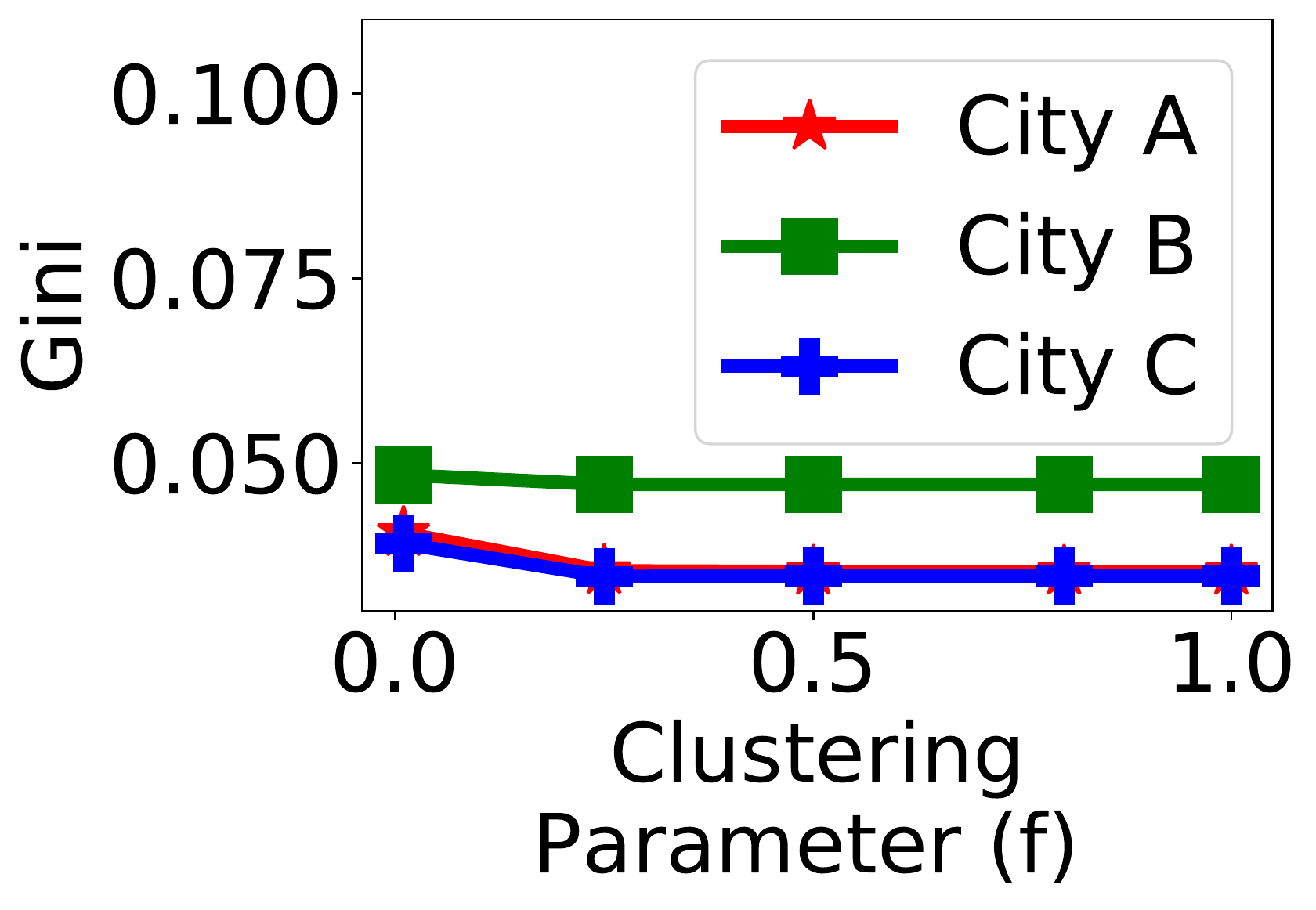}}
\subfigure{\label{fig:clustering_variation_income}\includegraphics[width=0.24\linewidth]{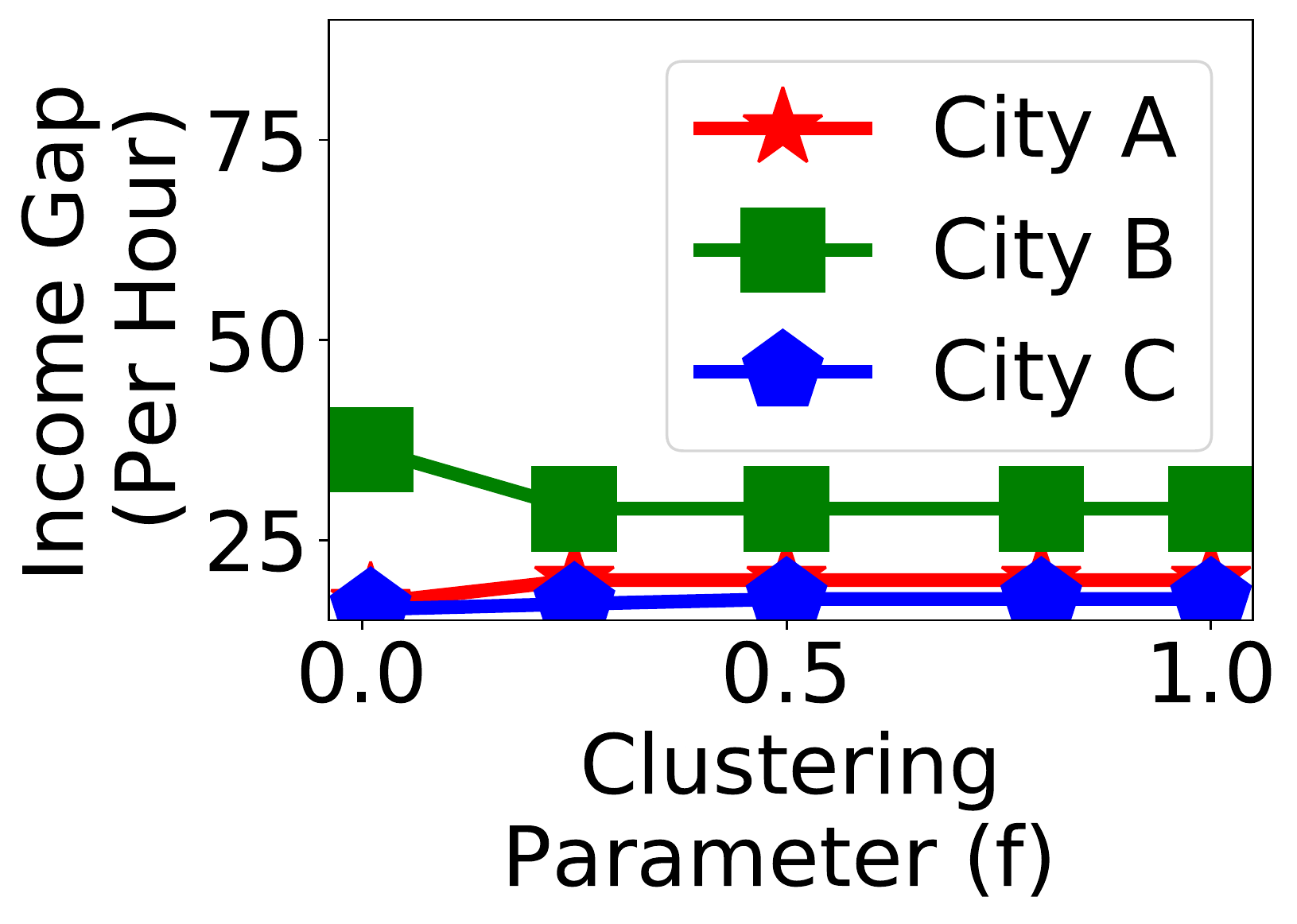}}
\subfigure{\label{fig:clustering_variation_dtpo}\includegraphics[width=0.24\linewidth]{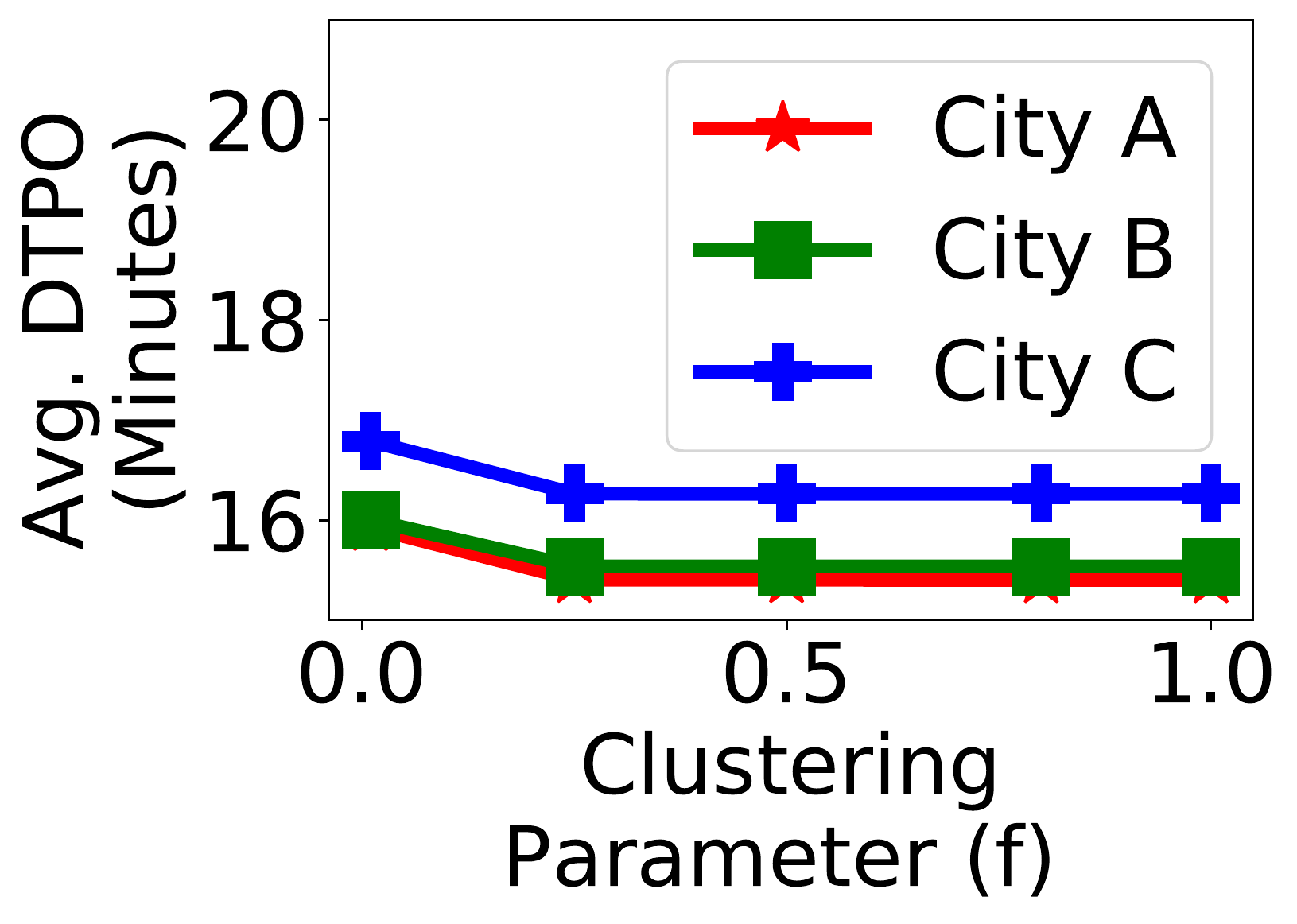}}
\subfigure{\label{fig:clustering_variation_sla}\includegraphics[width=0.24\linewidth]{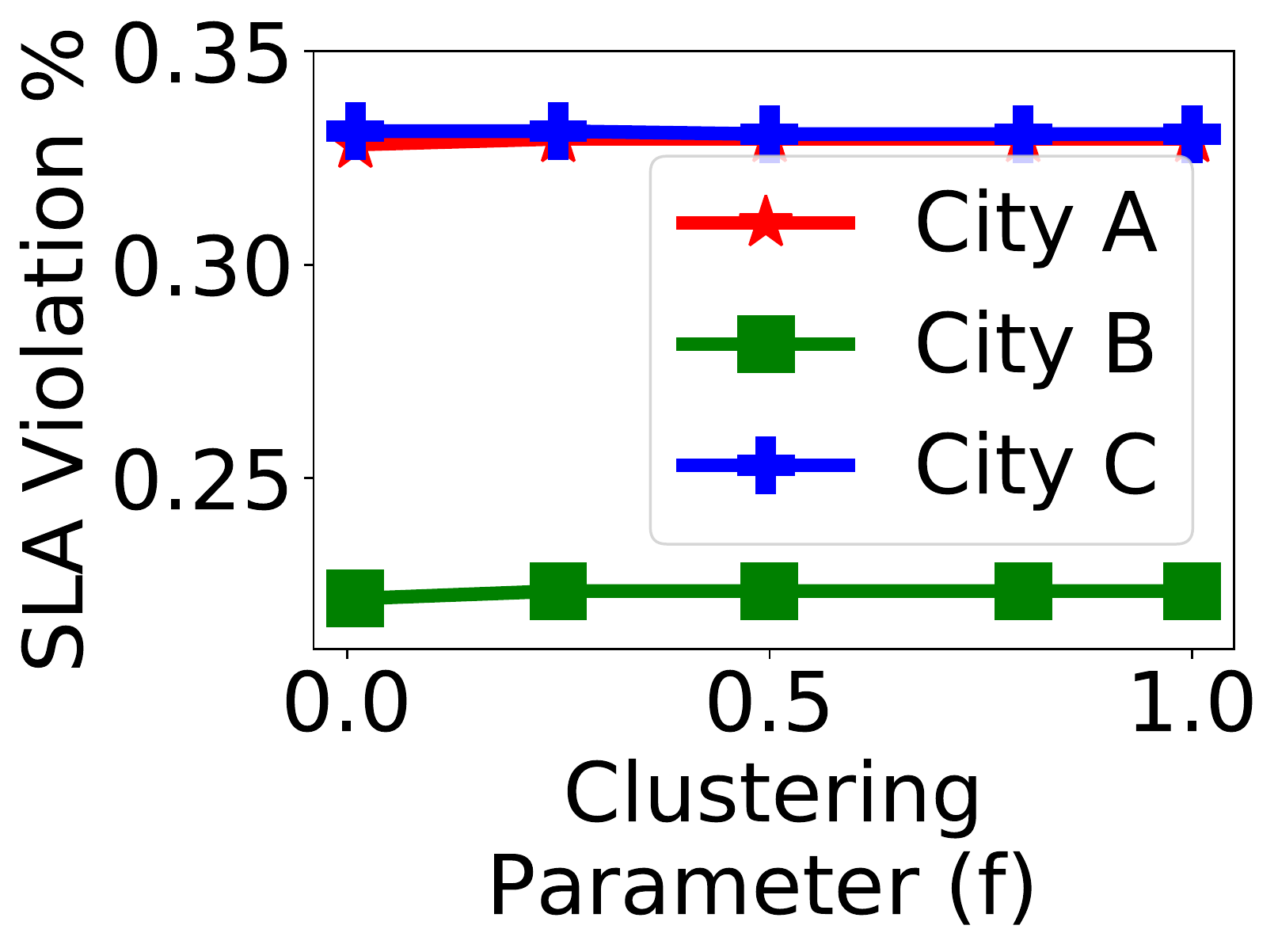}}
\caption{Impact of clustering parameter (f)}
\label{fig:clustering_variation}
\end{figure*}
\begin{figure*}[t]
\centering  
\subfigure{\label{fig:vehicle_variation_gini}
\includegraphics[width=0.23\linewidth]{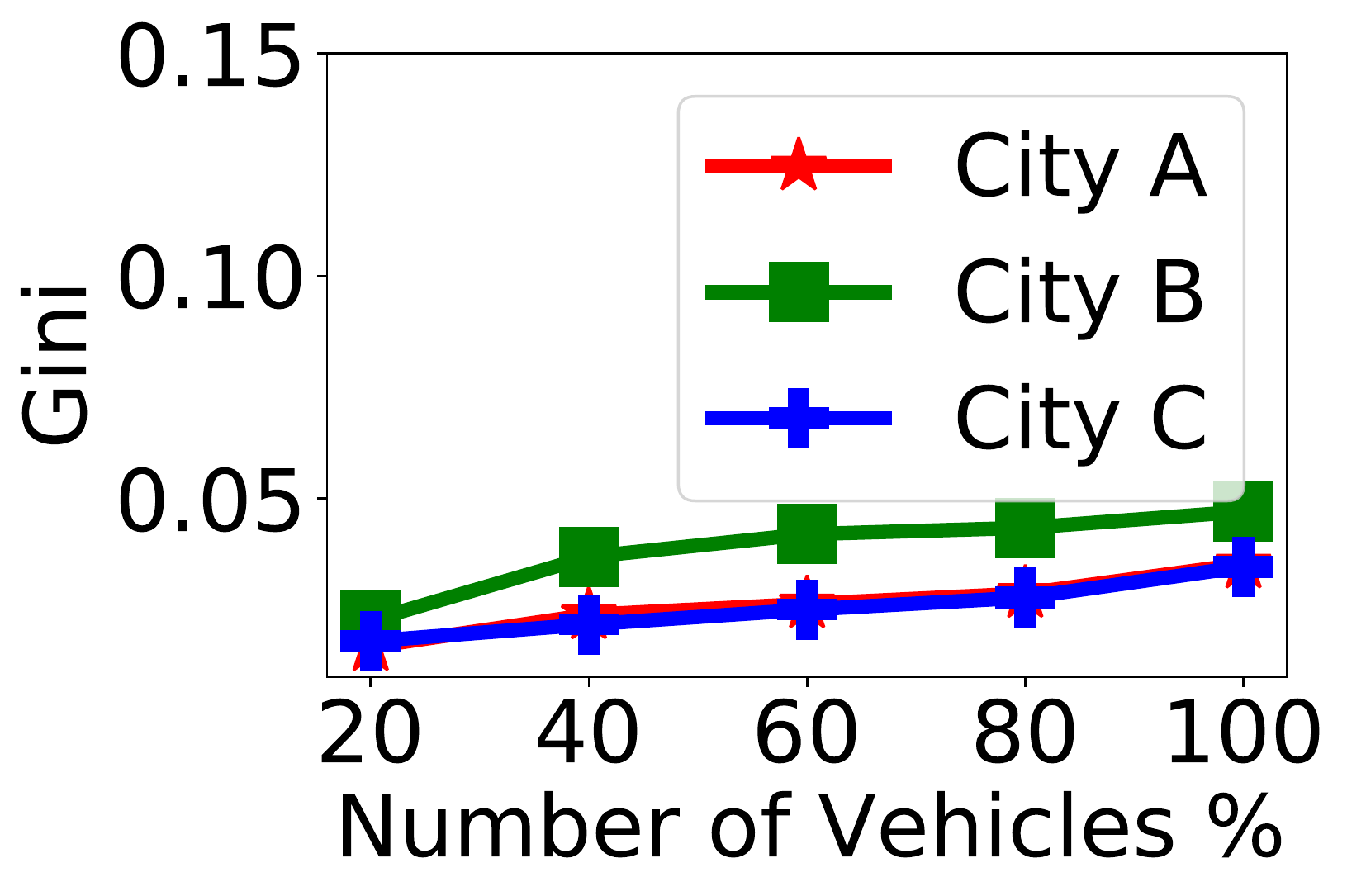}}
\subfigure{\label{fig:vehicle_variation_income}
\includegraphics[width=0.23\linewidth]{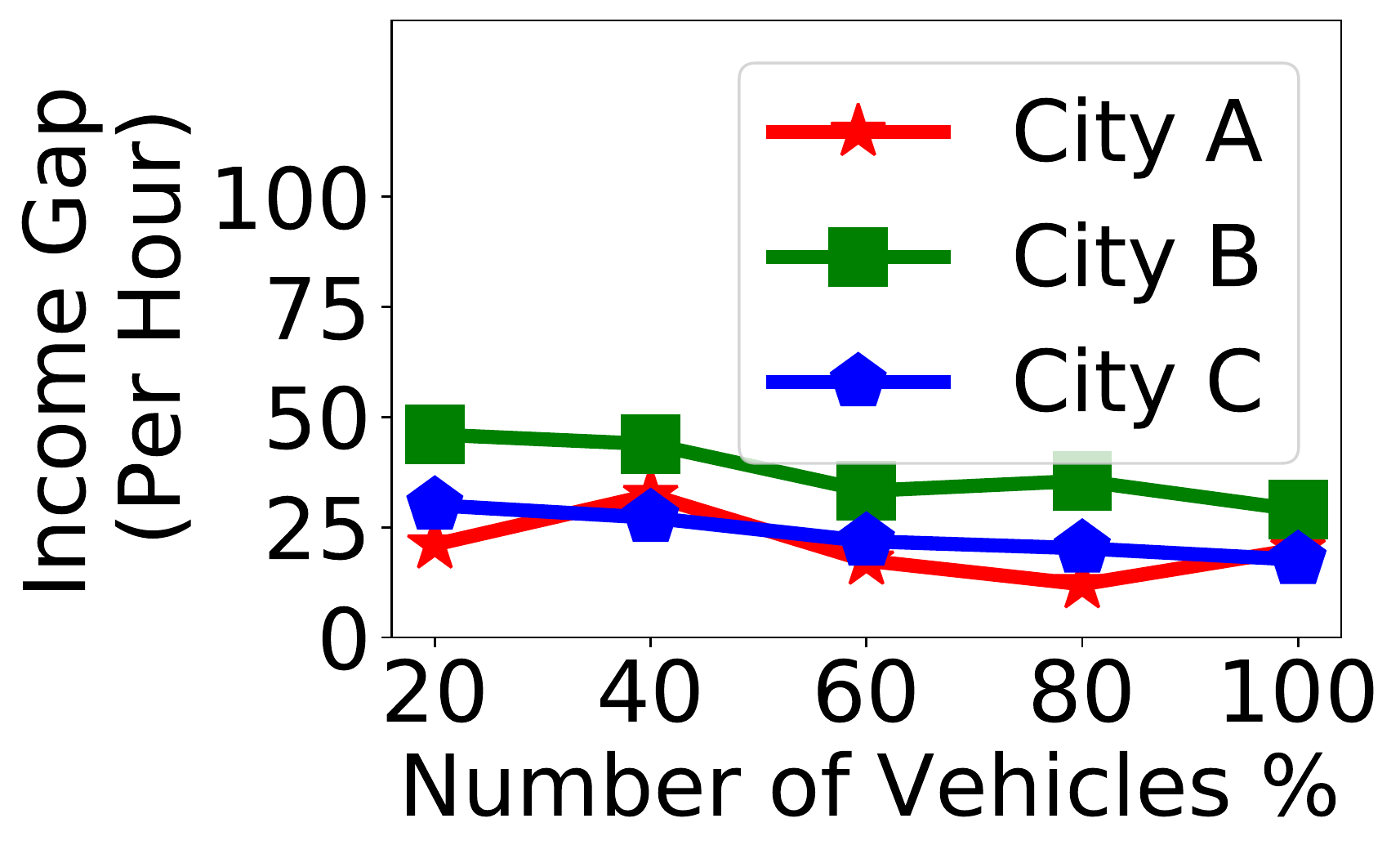}}
\subfigure{\label{fig:vehicle_variation_dtpo}
\includegraphics[width=0.23\linewidth]{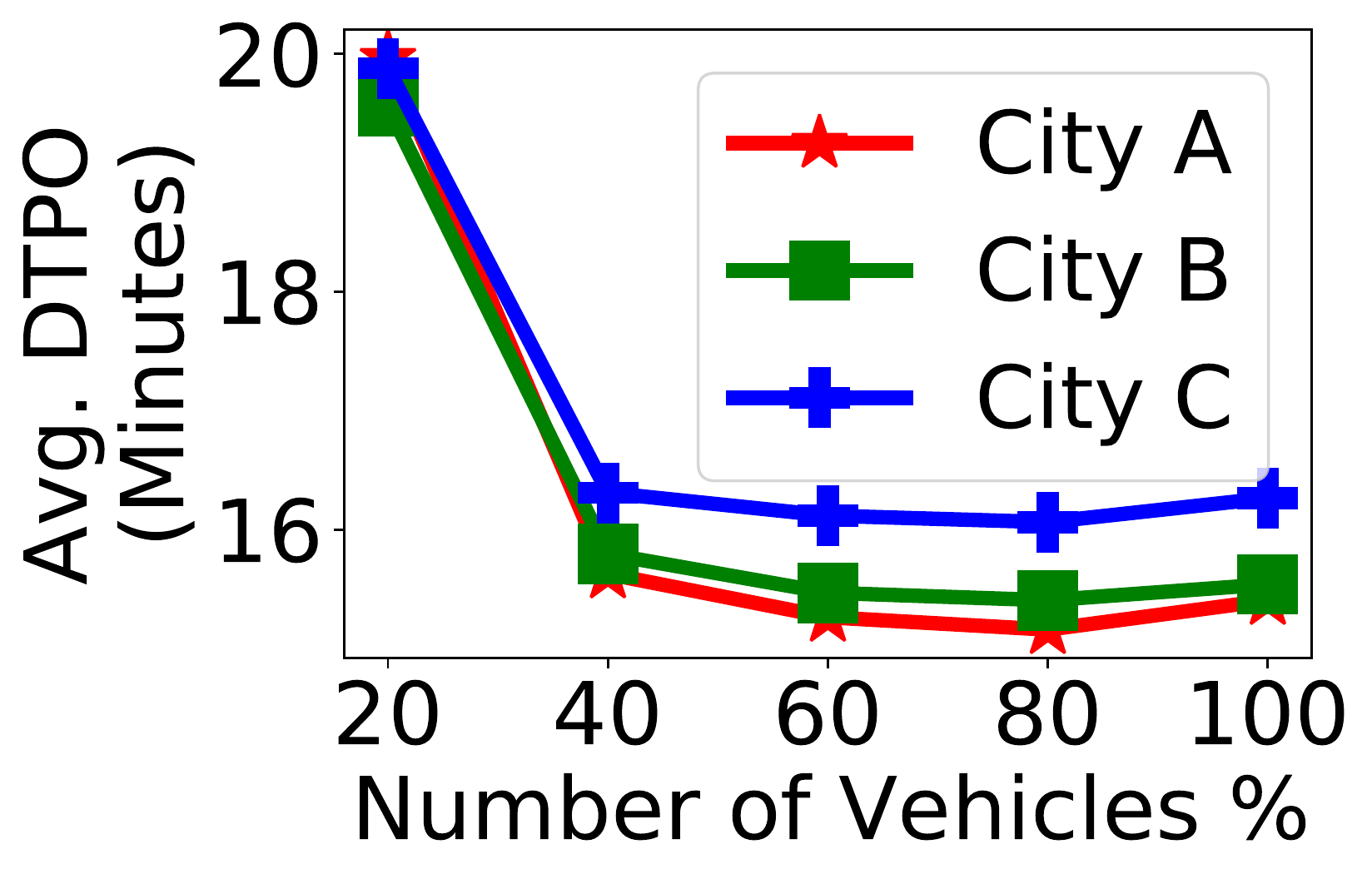}}
\subfigure{\label{fig:vehicle_variation_sla}
\includegraphics[width=0.23\linewidth]{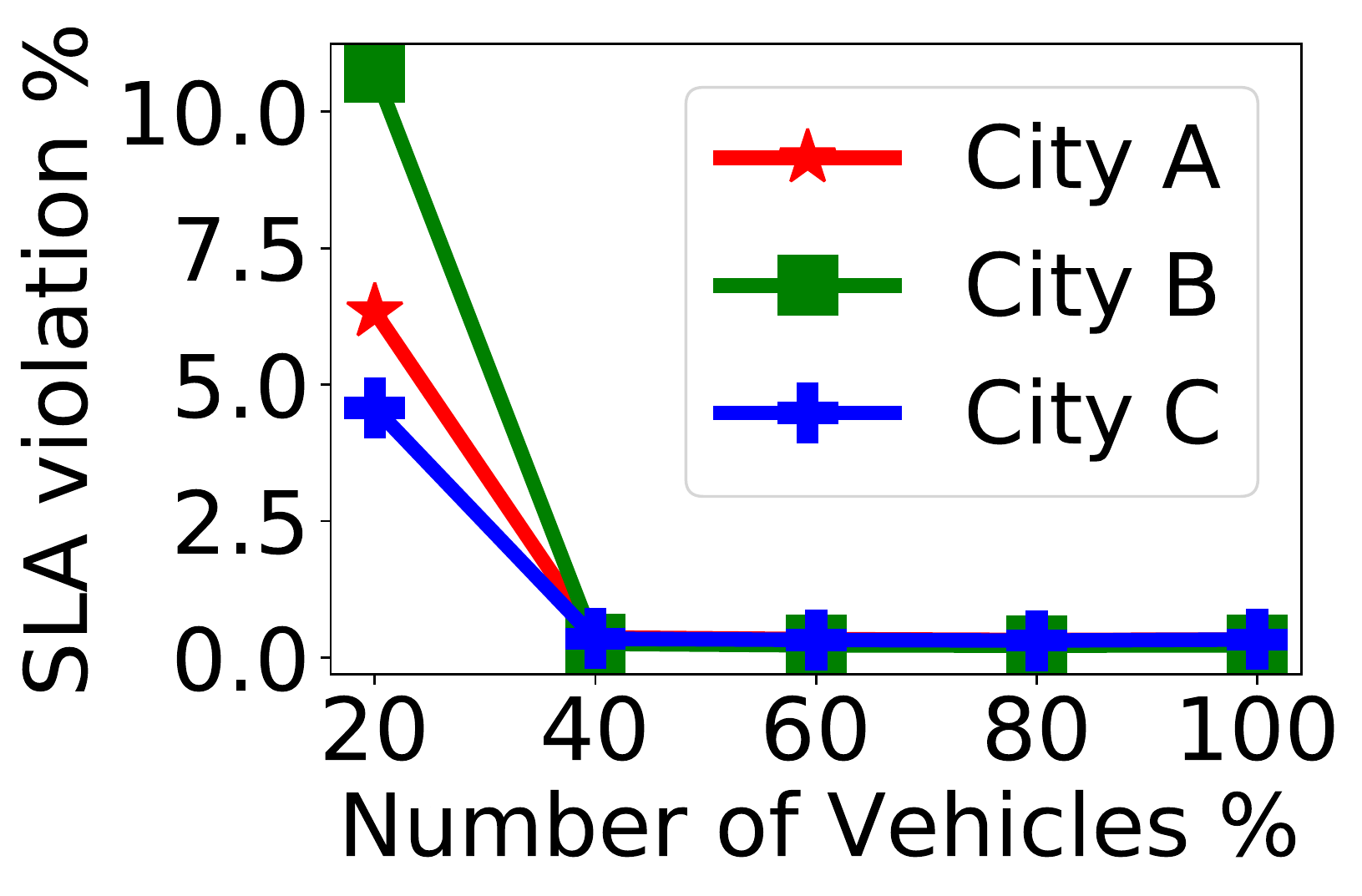}}
\caption{Impact of number of vehicles on various performance metrics.}
\label{fig:vehicle_variation}
\end{figure*}

\subsection{Scalability}
\label{section:scale}
\noindent
In food delivery, since we partition the data stream into windows of length $\Delta$, it is imperative that the allocation time of all orders within this window is smaller than $\Delta$. Otherwise, a queue will build up. We call a window ``overflown'' if the allocation time exceeds $\Delta$. Table~\ref{tab:runtime}, presents the percentage of overflown windows as well as average allocation time across all windows. As visible, all algorithms have $0$ overflown windows and hence are efficient enough for real-world workloads. Nonetheless, \fair has the smallest allocation time among all algorithms. 
\subsection{Impact of Parameters}
\label{sec:parameters}
\noindent
\textbf{Clustering parameter ($f$):}
Lowering of $f$ leads to more clustering among orders. Fig.~\ref{fig:clustering_variation} analyzes variation in different performance metrics against $f$. Both Gini coefficient (Fig.~\ref{fig:clustering_variation_gini}) and Average delivery time per order (Avg. DTPO) (Fig.~\ref{fig:clustering_variation_dtpo}) increase as we decrease $f$. Higher clustering prevents equitable distribution of orders (and, therefore, income). Hence, Gini increases. DTPO increases since larger clusters lead to less simultaneous delivery of orders through multiple drivers. SLA violations remain unaffected Fig.~\ref{fig:clustering_variation_sla}). There is no consistent trend across three cities in income gap Fig.~\ref{fig:clustering_variation_income}).

\noindent
\textbf{Impact of number of vehicles:}  Instead of considering all vehicles that were available in our real dataset, we randomly sample a subset of $X$ vehicles. Next, we vary $X$ in the $x$-axis and observe the impact of various metrics in $y$-axis. Fig.~\ref{fig:vehicle_variation} presents the results. We observe that Gini increases with increase in vehicles (Fig.~\ref{fig:vehicle_variation_gini}). This is expected since decreasing the number of vehicles generates higher demand and more scope to fairly distribute orders. It is also natural that higher vehicle available leads to lower delivery time (Fig.~\ref{fig:vehicle_variation_dtpo}) and less SLA-violations (Fig.~\ref{fig:vehicle_variation_sla}).\\
\subsection{Impact of traffic delays on fairness}
{\color{black}To see the impact of traffic delays on fairness, we synthetically increase the edge weight (travel time) by $50\%$ on $30\%$ of the edges chosen uniformly at random causing unexpected delays on our expected delivery times. Table~\ref{tab:delay_apdx} presents the impact on Gini. While we observe an increase in Gini compared to \fair, the Gini remains low and significantly better than \fmplus. The small increase in Gini happens due to our travel time estimations being violated.}\\
\subsection{Impact of fairness  on drivers' retention}
 {\color{black}Table.~\ref{tab:retention_apdx1} presents the percentage of drivers whose income increased under fair allocation through \fair when compared to their original allocation. As visible, $\approx 75\%$ drivers witness an increase in income. Hence, with fair allocation, retention gets easier.}  
 
 {\color{black}\textit{Do the income-decreasing drivers have higher skill? Will they leave the system due to fair allocation?} }
 {\color{black}
The skill level of a driver can be quantified through \textit{extra delivery time (XDT)} = delivery time - shortest possible delivery time (Def.~\ref{def:sdt}). In Table~\ref{tab:retention_apdx2}, we present the average difference in XDTs of income-increasing drivers under fair allocation with those whose income goes down ($\approx$ top-$25\%$) from the real food-delivery data shared by the service provider. As visible, the difference is minimal indicating similar skill-levels. This results substantiates our earlier claim that higher income is more closely linked with driver location rather than skill level (Fig.~\ref{fig:foodmatch_heatmap} and Table~\ref{tab:heatmap}). }

\begin{table}[t]
    \centering
    \scalebox{0.9}{
    \begin{tabular}{|l|l|p{0.7in}|p{0.8in}|p{0.7in}|}
    \hline
        \textbf{City} & \textbf{Algorithm} & \textbf{ Running time (secs)} & \textbf{Overflow (\%)} & \textbf{Peak Overflow (\%)} \\
      \hline
        \multirow{6}{*}{\textbf{A}} &\fair & \textbf{0.92} & 0 & 0 \\ 
       & \fmplus & 4.77 & 0 & 0 \\
       & \abhi $\lambda=1$ & 2.46 & 0 & 0 \\
        & \abhi $\lambda=0$ & 2.63 & 0 & 0 \\
         \hline
        \multirow{6}{*}{\textbf{B}}  & \fair & \textbf{14.00} & 0 & 0 \\
       & \fmplus & 41.93 & 0 & 0 \\
       & \abhi $\lambda=1$ & 21.00 & 0 & 0 \\
        & \abhi $\lambda=0$ & 20.67 & 0 & 0 \\
         \hline
        \multirow{6}{*}{\textbf{C}} & \fair & \textbf{10.76} & 0 & 0 \\
        & \fmplus & 50.66 & 0 & 0 \\
        & \abhi $\lambda=1$ & 28.21 & 0 & 0 \\
        & \abhi $\lambda=0$ & 26.12 & 0 & 0 \\
        \hline
    \end{tabular}}
    \caption{Comparison of running time and overflow windows. Peak Overflow considers only the lunch (11AM-2PM) and dinner (7PM-11PM) periods.}
    \label{tab:runtime}
\end{table}

\begin{table}[h]
    \centering
    \scalebox{0.9}{
    \begin{tabular}{c | c c c }
    \toprule
        \textbf{City} & \multicolumn{3}{c}{\textbf{Gini}}\\  
         & \fair & \fair + delay & FoodMatch \\
      \midrule
        \textbf{City A} & 0.035&0.08&0.518\\
        \textbf{City B }& 0.047&0.088&0.512\\
        \textbf{City C} & 0.035&0.067&0.562\\
         \bottomrule
    \end{tabular}}
    \caption{Impact of travel delays on Gini.}
    \label{tab:delay_apdx}
\end{table}
\noindent
\begin{table}[t]
    \centering
    \scalebox{0.9}{
    \begin{tabular}{c c c c c c c}
    \toprule
         \textbf{City} & \textbf{D1} & \textbf{D2} & \textbf{D3} & \textbf{D4} & \textbf{D5} & \textbf{D6}  \\
        \midrule
    \textbf{City A}&$74\%$& $75\%$& $73\%$&$73\%$& $71\%$& $72\%$\\
    \textbf{City B}&$74\%$& $75\%$& $73\%$&$74\%$& $72\%$& $72\%$\\
    \textbf{City C}&$78\%$& $79\%$& $77\%$&$78\%$& $77\%$& $76\%$\\
         \bottomrule
    \end{tabular}}
    \caption{Percentage of drivers whose income increased under \fair in each day (D1-D6).}
    \label{tab:retention_apdx1}
\end{table}

\begin{table}[t]
    \centering
    \begin{minipage}{2in}
    \scalebox{0.9}{
    \begin{tabular}{c c}
    \toprule
         \textbf{City} & \textbf{XDT Diff} (secs) \\
        \midrule
    \textbf{City A}&$11$\\
    \textbf{City B}&$44$\\
    \textbf{City C}&$68$\\
         \bottomrule
    \end{tabular}}
        \end{minipage}
    \caption{Performance of drivers in terms of delivery times.}
    \label{tab:retention_apdx2}
\end{table}

\section{Related Work}
\label{sec:relatedwork}
\noindent

\vspace{-1mm}
\noindent {\bf Food order assignment: } 
On the problem of food-delivery, \fmplus~\cite{foodmatch} is the only work to provide a realistic and scalable solution in food delivery domain. Other works on food delivery suffer from various unrealistic assumptions such as perfect information about arrival of orders~\cite{article:mealDelivery}, ignoring the road network~\cite{mdrp}, and ignoring food preparation time~\cite{10.14778/3368289.3368297}.

\vspace{1mm}
\noindent {\bf Fairness in multi-sided platform algorithms: }
With the growing popularity of multi-sided platforms, a number of recent works have investigated the challenges of unfairness and bias in such platforms. For example, \cite{article:racial} looked into the likelihood of racial bias in Airbnb hosts' acceptance of guests, while \cite{article:carrieraddtour} looked at gender discrimination in job advertisements. 
Few works have also looked at how producers and customers treat each other as a group. \cite{Chakraborty2017FairSF} and \cite{suhr2019two} proposed strategies for two-sided fairness in matching situations, whereas 
\cite{Burke} categorised distinct types of multi-stakeholder platforms and their required group fairness qualities. Individual fairness for both producers and customers is addressed by \cite{Patro_2020} in tailored suggestions in two-sided platforms. 
Despite these works on fairness in two-sided platforms, there has not been any studies on food delivery platforms. It is also worth noting that, as discussed in ~\cite{foodmatch_arxiv}, allocation algorithms for the cab service industry~\cite{mdm,yuen-www:2019,rp_insert1,rp_insert2} is not a natural fit food delivery.
\section{Conclusion}
\noindent
In this work, we focused on the unfairness issues faced by  delivery agents in food delivery platforms. Using data from a large real-world food platform, we showed that there exists high inequality in the income earned by the agents. To counter such inequality, we proposed an algorithm \fair to assign delivery agents to orders ensuring that income opportunities are {\it fairly distributed} among the agents. By removing zone restriction, \fair addresses the fact that the spatial spread of orders is a key driver of unequal pay. 
\fair also achieves a more fair pay distribution by amortizing fairness over a reasonable period of time, thereby ensuring that agents who rely only on food delivery for their livelihood are fairly remunerated. 
Extensive experiments show that \fair outperforms state-of-the-art baselines in lowering inequality while ensuring minimal increase in delivery time. Given the increasing adoptions of such platforms, it is the need of the hour and we hope that our work would lead to more followup works in this space.

\textbf{Limitations and Future Work:} Fairness can be studied from other angles as well. For example, do customers all over the city suffer equal extra delivery times? How does the presence of delivery vehicles in the neighborhood of a restaurant affect their order volumes? We plan to study these questions in the future.


\bibliography{mainAAAI}
\newpage
\end{document}